\documentclass[10pt]{amsart}
\usepackage{graphicx, mathabx, amssymb,amsfonts,amsmath,amsthm,newlfont}
\usepackage{epsfig,url}
\usepackage[usenames,dvipsnames]{color}
\usepackage{enumerate}
\usepackage[colorlinks=true,linkcolor=red,citecolor=blue]{hyperref}
\usepackage{color}

\usepackage[all,2cell]{xy} \UseAllTwocells \SilentMatrices

\newtheorem{thm}{Theorem}[section]
\newtheorem{cor}[thm]{Corollary}
\newtheorem{lem}[thm]{Lemma}
\newtheorem{prop}[thm]{Proposition}

\theoremstyle{definition}
\newtheorem{defn}[thm]{Definition}

\newtheorem{example}[thm]{Example}

\theoremstyle{remark}
\newtheorem{rem}[thm]{Remark}

\numberwithin{equation}{section}

\newcommand{\To}{\longrightarrow}

\newcommand{\Z}{\mathbb Z}
\newcommand{\Q}{\mathbb Q}

\newcommand{\C}{\mathbb C}

\newcommand{\R}{\mathbb R}
\newcommand{\N}{\mathbb N}
\newcommand{\Pro}{\mathbb P}

\newcommand{\q}{/\!/}

\newcommand{\tr}{\mathrm{tr}}

\newcommand{\can}{\mathrm{can}}

\newcommand{\cq}{\mathsf{q}}

\newcommand{\GL}{\mathrm{GL}}

\newcommand{\cone}{\mathsf{s}}
\newcommand{\trop}{\mathrm{trop}}

\newcommand{\Qt}{\mathcal{Q}}
\newcommand{\Qu}{\mathsf{1}}
\newcommand{\Qi}{\mathsf{i}}
\newcommand{\Qj}{\mathsf{j}}
\newcommand{\Qk}{\mathsf{k}}

\newcommand{\Ocan}{\mho}
\newcommand{\CC}{\mathcal{C}}

\newcommand{\Top}{\mathcal{T}op}

\begin{document}

\author{Francis Brown}
\begin{title}[Canonical Feynman integrals]
{Generalised graph  Laplacians and canonical Feynman integrals with kinematics}\end{title}
\maketitle

\begin{abstract} To any graph with external half-edges  and internal masses, we associate canonical  integrals  which   depend non-trivially on particle masses and momenta, and are always finite. They  are generalised Feynman integrals which satisfy graphical relations obtained from contracting edges in graphs, and  a coproduct involving  both ultra-violet and infra-red subgraphs.
Their  integrands   are defined by evaluating bi-invariant  forms, which represent stable classes in the  cohomology of the general linear group,  on a generalised graph Laplacian  matrix which depends on the external kinematics of a graph.
\end{abstract}

\section{Introduction}
In the paper \cite{CanonicalForms}  we introduced canonical differential forms on moduli spaces of metric graphs, and showed how they provide a connection between the  cohomology of the commutative graph complex, the algebraic $K$-theory of the integers, and Feynman integrals.
In this paper, we extend  this theory to the case of graphs with external momenta and masses, with an emphasis on  physical aspects.

Let us first recall some background on Feynman integrals. 
Consider a connected graph $G$, with $n$ external legs (or half-edges), $h_G$ loops, and $e_G$  internal edges. Every external leg represents an incoming particle with momentum $q_i \in \R^d$ subject to overall momentum conservation $\sum_{i=1}^n q_i=0$.  The number of spacetime dimensions $d$ will  be $d=2$ or $d=4$ in this paper.  To each edge $e$ one additionally associates one  of a finite set of particle masses  $m_e$, which are arbitrary real numbers.

For scalar theories,  the parametric Feynman integral  is the projective integral
\begin{equation} \label{FeynmanInt} 
I^{\mathrm{Feyn}}_G (q,m) =  \Gamma \left(e_G - \frac{h_G D}{2} \right)  \int_{\sigma_G}  \frac{ 1 }{\Psi_G^{D/2}}   \left(  \frac{ \Psi_G}{\Xi_G(q,m)}  \right)^{e_G-h_GD/2}  \Omega_G 
\end{equation} 
where the edges of $G$ are numbered from $1$ to $e_G$,  the form  $\Omega_G$ is defined by 
\[
\Omega_G  = \sum_{i=1}^{e_G}  (-1)^i \alpha_i d \alpha_1 \wedge \ldots \wedge \widehat{ d \alpha_i} \wedge \ldots \wedge  d \alpha_{e_G}\ ,
\]
and  $\sigma_G = \{ (\alpha_1: \ldots : \alpha_{e_G} ) \in \Pro^{e_G-1}(\R) : \alpha_i \geq 0\}$ is the coordinate simplex in projective space.  The quantity $D $  is typically an even integer,  or, in the setting of dimensional regularisation,  $D= 2k - \varepsilon$, for a small positive  $\varepsilon$.

The integrand involves the `second Symanzik' polynomial
\[ \Xi_G(q,m)  =  \Phi_G(q) + \left(\sum_{e \in E_G} m_e^2 \alpha_e\right) \Psi_G \]
which is expressed in terms of  the two more basic polynomials $\Psi_G, \Phi_G(q)$, which are homogeneous in the $\alpha_e, e\in E_G$ of degrees $h_G$, $h_G+1$ respectively. They are defined as sums over spanning forests in the graph $G$ with 1 or 2 connected components (see \S\ref{sect: DefSymanzik}). The  integral \eqref{FeynmanInt} diverges in general and can be regularised in a variety of manners including, for example,  Laurent expansion in the parameter $\varepsilon$. 

For most quantum field theories of relevance for collider physics, one is led to consider a wider class of Feynman integrals, which in parametric form (see, e.g. \cite{Golz}) have the following general shape (omitting $\Gamma$-factors for simplicity):
\begin{equation} \label{intro:Igeneral}
\int_{\sigma_G}  \frac{N_G}{\Psi_G^{a}\Xi^b_G(q,m)}   \Omega_G 
\end{equation} 
where $a,b\in \Z$ (or  $a,b \in \Z + \varepsilon\Z$), and the numerator $N_G$ is a polynomial in the parameters $\alpha_i$ with  typically complicated coefficients.  

Faced with the considerable difficulty in computing integrals  \eqref{FeynmanInt} or \eqref{intro:Igeneral},  a common recent theme  of research is to seek alternative theoretical frameworks in which the corresponding amplitudes are simpler and more highly structured, with a long term view to unearthing  mathematical  properties which are valid for  general quantum field theories.
Notable examples in this direction include the amplituhedron programme \cite{Amplituhedron}, which  organises certain amplitudes in $N=4$ SYM according to geometric principles; string perturbation theory, which studies scattering amplitudes defined on punctured Riemann surfaces; or integrable `fishnet'  models \cite{Fishnet, Gurdogan} which  reduce to a small number of Feynman graphs. 

In this paper, we introduce a special class of geometrically-defined  integrals of the form \eqref{intro:Igeneral}, which  are always finite irrespective of the graph $G$. 
They  generalize  \cite{CanonicalForms} to incorporate masses and momenta, and  satisfy  symmetry properties including a  family of  graphically-encoded relations.  
  The latter  involve both contraction of internal edges  and   also the `motic' coproduct of \cite{Cosmic} which has applications to the study of both UV and  IR divergences \cite{IR, IRUVtropical}.

\subsection{Canonical forms and their integrals} Our starting point is the moduli space  $\mathcal{M}^{\trop}_{g}$ of stable tropical curves of genus $g$, or more precisely the open locus 
\[
\left(\mathcal{M}^{\trop}_{g}\right)_{w=0}   \subset  \mathcal{M}^{\trop}_{g} \]
  consisting of graphs whose vertices have  weight zero.  The quotient $\left(\mathcal{M}^{\trop}_{g}\right)_{w=0} /\R^{\times}_{>0}$ is isomorphic to the quotient of Culler-Vogtmann's Outer space $\mathcal{O}_g$  \cite{CullerVogtmann} by the group of outer automorphisms $\mathrm{Out}(F_g)$ of the free group $F_g$ on $g$ elements. 
      Points in this space are equivalence classes of connected metric graphs (with no external legs), where a metric on a graph is an assignment of  a  length $\ell_e>0$ to  every internal edge, normalised  so that the total length of all edges  is $1$.  
      
      The  tropical Torelli map \cite{TropicalTorelli} on $\mathcal{M}^{\trop}_{g}$ restricts to a map 
      \begin{equation} \label{tropg} 
\left(\mathcal{M}^{\trop}_{g}\right)_{w=0} \To  \mathcal{P}_g/\GL_g(\Z)
\end{equation} 
where $\mathcal{P}_g$ is the space of positive definite $g\times g$ symmetric matrices $X$, upon which $ P \in \GL_g(\Z)$ acts via $X \mapsto P^T X P$. 
The map \eqref{tropg} assigns to 
a connected  metric graph $G$ the $\GL_g(\Z)$-equivalence  class of any choice of graph Laplacian matrix  $\Lambda_G$. The latter is  a  $g \times g$ symmetric matrix   whose determinant 
\[  
\det \Lambda_G = \Psi_G
\]
 equals the first Symanzik, or  Kirchhoff, graph polynomial which arises in the integrals \eqref{FeynmanInt}.  The graph Laplacian
 may be interpreted as a tropical version of the Riemann polarization form on a compact Riemann surface of genus $g$. 
 Canonical forms on $\left(\mathcal{M}^{\trop}_{g}\right)_{w=0}$ are defined as follows.
 The  forms 
$\tr( (X^{-1} d X)^{4k+1})$, for $k\geq 1$, are bi-invariant under  left and right multiplication by   $\GL_g(\R)$ and were shown by Borel to 
 generate the stable cohomology of the symmetric space $\mathcal{P}_g/\GL_g(\Z)$.  Their pull-backs  $\tr( (\Lambda_G^{-1} d \Lambda_G)^{4k+1})$  along the map \eqref{tropg} defines a distinguished family of differential forms on 
the images of cells in $\left(\mathcal{M}^{\trop}_{g}\right)_{w=0} /\R^{\times}_{>0}$. The integrals of these forms are called canonical since they do not depend on any choices, and were studied in \cite{CanonicalForms}. In that paper it was shown that interesting examples of Feynman residues for vacuum diagrams (see \cite{SchnetzCensus} for a survey) of relevance to quantum field theory can arise as canonical integrals. A natural question  is whether general  Feynman integrals with non-trivial kinematics are also amenable to a `canonical' geometrical interpretation of this kind.  
This is   addressed in this paper.

\subsubsection{Canonical forms with kinematic dependence}  We consider  a   generalised moduli space $\mathcal{M}^{\trop}_{g,n,M}$, first considered in \cite{ModuliColoured}, of weighted metric graphs which have $n$ external legs labelled from $1,\ldots, n$, and where each  internal edge $e$ is assigned one of $M$ possible particle  masses. 
 When  $M=0$,  this space coincides with the moduli space  $\mathcal{M}^{\trop}_{g,n}$ of $n$-marked tropical curves of genus $g$. Denote the open locus of graphs whose  vertex weights are all zero by:
 \begin{equation}  
\left( \mathcal{M}^{\trop}_{g,n,M}  \right)_{w=0}  \subset \mathcal{M}^{\trop}_{g,n,M} \ . 
\end{equation}
Consider in the first instance  external particle momenta $q_1,\ldots, q_n \in \R^2$ in two-dimensional Euclidean space, which we identify with the complex numbers,  and any choice of 
$M$ internal particle masses. For any choice of routing $\mu_e$ of the external momenta through the internal edges $e$ of $G$,  we consider  a generalised graph Laplacian matrix  (see  \eqref{LambdaGblockmatrix}) whose  determinant satisfies 
\begin{equation} \label{intro: detLambdatilde} 
\det \widetilde{\Lambda}_G(\mu,m) = \Xi_G(q, m) \ . 
\end{equation} 
The generalised Laplacian  $\widetilde{\Lambda}_G(\mu,m)$ is a Hermitian $(g+1) \times (g+1)$ matrix (note that $g+1$ is the degree of the   polynomial  $\Xi_G(q,m)$). It is the tropical version of a regularised Hermitian polarization form on the cohomology of
a  compact Riemann surface with $n$ punctures. From this perspective, the external particle momenta are interpreted as  tropical versions of the residues 
of a  differential form of the third kind.   The case when the external momenta lie in $\R^4$ may be treated by replacing complex momenta with  momenta in the ring of quaternions $\Qt \cong \R^4$, leading to a quaternionic Hermitian graph Laplacian which   satisfies
 \begin{equation} \label{intro: detLambdaQuattilde} 
\mathrm{Det} \, \widetilde{\Lambda}_G(\mu,m) = \left(\Xi_G(q, m)\right)^2 \ ,  
\end{equation}
where $\mathrm{Det}$ of a quaternionic matrix is the determinant of its complex adjoint. 
 One of our key results (\S\ref{sect: Canatinfinity}) is  that the generalised graph Laplacian admits an asymptotic decomposition into block matrices as one approaches the boundary of the moduli space $(\mathcal{M}^{\trop}_{g,n,M})_{w=0}$.   This provides a geometric  interpretation of  asymptotic factorisation formulae for graph polynomials which are important for the study of infra-red and ultra-violet singularities of Feynman integrals \cite{IR, IRUVtropical}.

Define  canonical forms\footnote{The reader familiar with Chern-Simons theory may recognise canonical forms as a special case: for example, the Chern-Simons 3-form 
  $\tr( dA \wedge A +\frac{2}{3} A\wedge A \wedge A)$ reduces to $-\frac{1}{3} \tr((\Lambda^{-1} d\Lambda)^3)$ upon substituting  $A=\Lambda^{-1} d\Lambda$. The higher degree forms are similar.} of the `first' and `second' kinds  for all $k \geq 1$, by 
 \begin{eqnarray} \label{intro: canonical} 
 \omega^{4k+1}_G & = &  \tr\left( \left( \Lambda_G^{-1} d\Lambda_G \right)^{4k+1}  \right)  \\ 
 \varpi^{2k+1}_G & = &  \tr\left( \left( \widetilde{\Lambda}_G(\mu,m) ^{-1} d  \widetilde{\Lambda}_G(\mu,m)  \right)^{2k+1}  \right)  \  ,  \nonumber
  \end{eqnarray} 
  when the external momenta lie in $\R^2 \cong \C$. In the case of quaternionic momenta in $\R^4 \cong \Qt$, the definition of  canonical forms of the first kind is unchanged, but the canonical forms of the second kind are now  defined for all $k \geq 1$ by: 
  \begin{equation} \label{intro: canonicalQuat}  
    \varpi^{\Qt, 4k+1}_G  =   \mathrm{Tr} \left( \left( \widetilde{\Lambda}_G(\mu,m) ^{-1} d  \widetilde{\Lambda}_G(\mu,m)  \right)^{4k+1}  \right)  \  .  
   \end{equation} 
   where the quaternionic trace $\mathrm{Tr}$ is the trace of the complex adjoint.
  A crucial property of canonical forms (theorem \ref{thm: Canformsalongdivisors})  is that they  factorise as one approaches infinity along the boundary faces of Feynman polytopes.

   When the external momenta lie in $\R^2 \cong \C$, we  define a \emph{canonical form}  $\omega_G$  
    to be any  homogeneous polynomial in the forms \eqref{intro: canonical}. 
    When the momenta lie in $\R^4 \cong \Qt$,  a  (quaternionic) canonical form, denoted $\omega^{\Qt}_G$,
    is a homogeneous polynomial in $\omega^{4k+1}_G$ and $ \varpi^{\Qt, 4k+1}_G$.  For any connected graph $G$, and  any  such form $\omega_G$ of degree $e_G+1$, we may consider the integral
\begin{equation}  \label{intro: canonicalint}
I_G(\omega, q, m) =  \int_{\sigma_G}  \omega_G \ . 
\end{equation} 
Since it is independent of all choices which go into defining  generalised graph Laplacian matrices, we call it  a \emph{canonical integral}.

We prove the following facts about the integrals \eqref{intro: canonicalint}:

\begin{enumerate}
\item They are always finite.
\item They are generalised Feynman integrals of the form \eqref{intro:Igeneral}.
\item They satisfy graphical relations  involving  both the contraction of internal edges   and a generalisation of the Connes-Kreimer coproduct which encodes not only ultraviolet, but also infra-red phenomena. 
\end{enumerate}

Thus the integrals  \eqref{intro: canonicalint}, which may be loosely  interpreted as certain `volumes'  of cells on a moduli space of tropical curves,  pick out a distinguished class of generalised Feynman integrals with special properties. The reader may wish to turn  to sections \ref{sect: ExamplesComplex} - \ref{sect: ExamplesQuaternionic}
for examples of canonical integrals and their relations. 
\begin{rem} The canonical forms  \eqref{intro: canonical}, \eqref{intro: canonicalQuat}  arise from   the graded exterior algebras of invariant differential forms $I_{H_n}$ on the symmetric spaces associated to the following sequences of classical non-compact Lie groups $H_n$. The following table is extracted from \cite{Borel}, page 265, and based on results of H. Cartan:
 \[
\begin{array}{cc|cc}
H_n   && & \underset{\longleftarrow}{\lim} \, I_{H_n}  \\ \hline 
\mathrm{SL}_n(\R)     &&  & E[ x_{4k+1},  k\geq 1 ]  \\
\mathrm{SL}_n(\C)   &&  &    E[ x_{2k+1},  k\geq 1 ]  \\
 \mathrm{SL}_n(\Qt)    &&  &   E[x_{4k+1},  k\geq 1 ]
\end{array}
\]
where  $E[\{x_i\}]$ denotes the graded exterior algebra generated by elements $x_i$ in odd degrees $i$, and the inverse limit  of the spaces of forms $I_{H_n}$ is taken in the category of graded exterior algebras. The generators are given by $\beta^n_X$ in the case of real and complex matrices, and by  $ \mathrm{Tr}  ((X^{-1} d X)^n)$ in the quaternionic case. They always vanish when $n$  is even, and    additionally for  $n\equiv 3 \pmod 4$ 
 in the case of $\R$ and  $\Qt$. 
  \end{rem}

\subsection{Discussion and questions for further research} 
It would be very interesting to find a momentum or position space formulation for the integrals \eqref{intro: canonicalint}, and to relate them more closely to the theory of graphical functions \cite{GraphicalParametric,BorinskySchnetz} in the case when all masses vanish. 
The generalised graph Laplacian is a tropical version of `single-valued' period integrals on a  punctured Riemann sphere, which are prevalent in closed string perturbation theory, and suggests a possible connection with string perturbation theory. In a different direction, it was shown in \cite{CanonicalForms} that the canonical integrals for graphs without   kinematics are closely connected to the homology of the commutative, even,  graph complex.  It would be interesting to find such a  relation for the integrals
\eqref{intro: canonicalint} based on corollary \ref{cor:StokesCompactType}, for example.
See also \cite{BerghoffKreimer} for related interpretations of Feynman integrals. 

We also expect the work in this paper to have applications to the motivic study of Feynman integrals. For example, the existence of the generalised graph Laplacian implies that the graph motives at a fixed loop order have  a universal family, and the canonical forms   provide  universally-defined classes in the  relative cohomology of graph hypersurface complements and the de Rham realisation of graph motives.

In particular, since canonical integrals are distinguished elements of period matrices associated to ordinary Feynman integrals, they should be related to the latter via differential equations and also via the motivic coaction. The Stokes' relations proven here should thus enable one to transfer information about differential equations, periods, and infra-red singularities between graphs with different topologies. 

Finally, it would also be interesting to study integrals    
 of the form 
 \[ \int_{\sigma_G} \left(  \frac{\Psi^{h_G+1}_G}{\Xi^{h_G}_G(q,m)}\right)^{\varepsilon} \omega_G \] in dimensional regularisation, where $\omega_G$ is canonical. They too will satisfy relations via Stokes' formula, but  of a slightly different form to those considered here.
 
\subsection{Contents}
 In $\S2$ we recall some notations and conventions as well as some background on quaternionic matrices. In \S3 we describe a moduli space of marked metric graphs with additional edge colourings, which encode particle masses. This space was previously introduced by Berghoff and M{\"u}hlbauer in \cite{ModuliColoured}.
 After a brief discussion in \S\ref{sect: MomInterp}      of momentum routings, and complex and quaternionic momenta,  we define the generalised Laplacians in \S\ref{sec: GenLaplace}. In the case when momenta are quaternionic and all masses are zero, they  coincide with matrices previously considered by Bloch and Kreimer in \cite{BlochKreimer}. The way  we arrived at the definition is  explained in \S\ref{sect: TropicalSV}; namely by taking the tropical analogue of the single-valued integration pairing on the cohomology of a  punctured Riemann surface.
  In \S \ref{sect: Symanzik} we recall the definitions of Symanzik polynomials, and prove a formula for $\Phi_G$ in terms of spanning forest polynomials, which is used to deduce
 \eqref{intro: detLambdatilde} and \eqref{intro: detLambdaQuattilde}.
  Starting in \S\ref{section: invariantformsofgraphs}, we recall the definitions of bi-invariant differential forms. A key result is a formula (\S\ref{sect: Canatinfinity}) for  the asymptotic behaviour at infinity of graph Laplacian matrices,   which is used 
to prove the convergence of canonical integrals in  \S\ref{sect: CanInt}, along with the  Stokes relations. It could also be used to prove  \eqref{intro: detLambdatilde} and \eqref{intro: detLambdaQuattilde}.      Sections \ref{sect: ExamplesComplex} and    
 \ref{sect: ExamplesQuaternionic} provide examples of canonical integrals for some well-studied Feynman diagrams, and \S\ref{sect: Stokes rel} illustrates the general Stokes relations by deducing a graphical 5-term functional equation  for the massive box diagram in 4 spacetime dimensions.
 
  \vspace{0.1in}

\emph{Declarations.}
The author has no relevant financial or non-financial interests to disclose,  nor any competing interests to declare that are relevant to the content of this article. 
\\

\emph{Acknowledgements}. Many thanks to  Marko Berghoff, Paul Fendley, Lionel Mason,  Erik Panzer, and  all the participants of the seminar at Oxford in 2022 on graph complexes.   This project has received funding from the European Research Council (ERC) under the European Union's Horizon 2020 research and innovation programme (grant agreement no. 724638).

\section{Notations and conventions}
Graphs in this paper are finite,  and  usually connected. They have a finite number $n$  of external half-edges. The set of internal edges is denoted by $E(G)$, and the set of vertices by $V(G)$. We write $e_G = |E(G)|$ for the number of internal edges, and $h_G$ for the number of loops of $G$. 
When an internal edge is directed, we denote by $s(e)$ its source, and $t(e)$ its  target. An edge with $s(e)=t(e)$ is called a self-edge or tadpole.

\subsection{Reminders on Quaternions} \label{sect: Quaternions}
Denote the ring of quaternions by
\[
\Qt = \R \Qu+  \R \Qi +  \R \Qj +  \R \Qk  \ , 
\]
where $\Qu$ is the unit, $\Qi^2 = \Qj^2=\Qk^2=-\Qu$, and $\Qi \Qj=-\Qj \Qi= \Qk, \Qk \Qi=-\Qi \Qk= \Qj, \Qj \Qk=-\Qk \Qj=\Qi$. It admits a representation
$
\chi: \Qt  \rightarrow M_2(\C)$  via 
\[ 
\Qu \mapsto \begin{pmatrix} 1 & 0 \\ 0 & 1 \end{pmatrix} \quad , \quad 
\Qi \mapsto \begin{pmatrix} i & 0 \\ 0 & -i \end{pmatrix} \quad , \quad 
\Qj \mapsto \begin{pmatrix} 0 & -1 \\ 1 & 0 \end{pmatrix} \quad , \quad 
\Qk \mapsto \begin{pmatrix} 0 & -i \\ -i & 0 \end{pmatrix} \ . 
\]
The ring $\Qt$ has  an anti-involution,  denoted $x \mapsto \overline{x}$, which is induced  by Hermitian conjugation  $V \mapsto (\overline{V})^T$ on $M_2(\C)$. The conjugate of $
x = x_1 \Qu + x_2 \Qi + x_3 \Qj + x_4 \Qk 
$ is 
$
\overline{x} = x_1 \Qu  -  x_2 \Qi - x_3 \Qj - x_4 \Qk 
$, 
 and  the square of the quaternion norm 
\[ 
|\!|  x |\!|^2 =  x \overline{x}  =  (x_1^2 + x_2^2 +x_3^2 +x_4^2) \Qu
\]
coincides with the square of  the Euclidean norm.

\subsubsection{Quaternionic matrices}  \label{sect: QuatMatrices} Matrices with entries in the ring of quaternions share many properties with rings of matrices over  a commutative ring (see  \cite{QuaternionMatrices} for a survey).
A convenient method  for studying them is via  the complex adjoint representation, which is the ring homomorphism
\begin{equation} \label{chidef}  \chi :  M_n(\Qt) \To M_n(M_2(\C)) \ , 
\end{equation}
induced by the representation $\chi$ defined above. One shows that a matrix $M \in M_n(\Qt)$  has a unique left inverse (which is necessarily also a right inverse) if and only if  its image $\chi_M$ is invertible. However, the determinant of a quaternionic matrix is not  defined in  general, and its trace is not a similarity invariant, meaning that $\tr(P^{-1}MP) \neq \tr(M)$ for general $P,M$.   For this reason it is customary to define the determinant and trace of $ M \in M_n(\Qt)$ to be $\mathrm{Det} (M) = \det \chi_M$ and $\mathrm{Tr}(M) = \tr \,\chi_M$ respectively. They take values in the complex numbers.

Two particular representations of quaternionic matrices  are convenient: 
\vspace{0.05in}

\emph{(i)}.  If one identifies $M_n(M_2(\C)) $ with $M_{2n}(\C)$, then the 
 image of $\chi$  is the set of matrices $\{V \in M_{2n}(\C): \overline{V}=J_n V J_n^{-1}\}$,
where $J_n$ is the block diagonal matrix consisting of $n$ copies of $\chi_{\Qj}$ along the diagonal.
\vspace{0.05in}

\emph{(ii)}.  If  one  identifies $M_n(M_2(\C)) $ with $M_2(M_n(\C))$,  then the image of $\chi$ consists of block matrices of the form:
\begin{equation} \label{QuaternionicMatrices} 
\chi_{A + \Qj B} = \begin{pmatrix} A & B \\ - \overline{B} & \overline{A} \end{pmatrix} \ , \qquad A, B \in M_n(\C) \  , 
\end{equation} 
since  any quaternionic matrix may be uniquely written    $A + \Qj \,B$, with $A, B \in M_n(\C)$. Here, and later, we identify the complex numbers with 
$\R \Qu + \R \Qi \subset \Qt$. 
\vspace{0.05in}

Either representation may be used to compute the quaternionic trace or determinant. We shall use $(i)$ in the examples.

\subsubsection{Moore determinant}

The first part of the following proposition implies  the existence of  the reduced Pfaffian norm $\mathrm{Nrp}$ \cite{Moore, Tignol}.

\begin{prop} \label{prop: Hermmatrices} Let $\chi$ be the map  \eqref{chidef}. 
\vspace{0.05in}

 (i). Let  $M \in M_n(\Qt)$ be a quaternionic Hermitian matrix, i.e.,  $M^T = \overline{M}$, where $\overline{M}$ denotes its quaternionic conjugate. Then $\chi_M$ is Hermitian and
 \begin{equation} \label{detissquare} 
\mathrm{Det}(M) = \det( \chi_{M}) =  \left(F(\chi_{M})\right)^2\ ,
 \end{equation}
 where $F$ is a  polynomial in the entries of $\chi_{M}$ which  is invariant under complex conjugation. 
 \vspace{0.05in} 

(ii).  In the case when $M \in M_n(\C) \subset M_n(\Qt)$  has complex entries \begin{equation}
 \mathrm{Det}(M)= \det( \chi_M) =  \det(M)^2\ .
 \end{equation}
\end{prop} 
Although we shall not use it, the Moore determinant $F=\mathrm{Nrp}(M)$ is a canonical solution to \eqref{detissquare} which satisfies $\mathrm{Nrp}(M) = \det(M)$ for all $M \in M_n(\C)$.

\section{Moduli of marked metric graphs with masses and momenta} 
We consider a moduli space of metric graphs with external legs and masses.  It was previously studied in \cite{ModuliColoured}. 
\subsection{Graphs with external legs and massive edges} 

\subsubsection{Combinatorial masses}
Let $M\geq 0$ be a fixed integer. 
For a connected graph $G$, the assignment of a mass to each internal edge $e\in E(G)$ is encoded by
\[
 \underline{M} : E(G) \To \{0,1, \ldots, M\} \ .
\]
The labelling $0$ denotes the zero mass.

\subsubsection{Physical masses} We fix a finite set of non-zero masses $m_i\in \R^{\times}$, for $i=1,\ldots, M$, and set $m_0=0$ to be the zero mass, i.e., $m: \{0,1\ldots, M\} \rightarrow \R$.
 Only the  mass   squares  $m_i^2$ will play a role in the theory. The mass $m_i$ denotes the mass of a  particle of type $i$, and there is no need for the $m_i$ to be distinct. 
When drawing Feynman diagrams, our convention is to depict massless edges with a single line, and massive edges with a doubled line. 

Two cases are of particular interest, namely the `massless case'  $M=0$,  when all edges have mass zero;
 and the  `generic mass case', when  $M=e_G$,  and $\underline{M}: E(G) \overset{\sim}{\rightarrow} \{1,\ldots, e_G\}$  is a bijection.

\subsubsection{Combinatorial momenta}  Let $n\geq 0$ be a fixed integer.  If $n\geq 1$, the data of $n$ external legs or `markings'  is given by a map
\[ 
\underline{n} : \{1,\ldots, n\} \To V(G)
\]
which can be realised combinatorially by attaching a half-edge, labelled $i$, to the vertex $\underline{n}(i)$.  These external half-edges will always be oriented towards their endpoint, i.e., all external momenta are incoming. The map $\underline{n}$ does not need to be injective or surjective.
In the case when $n=0$,  or when the image of $\underline{n}$ reduces to a single vertex, the associated graph is called  a vacuum diagram, or  momentumless. 

\subsubsection{Physical momenta}
 Let $d \geq 0$. Consider $n$ vectors in Euclidean space
 \[
Q = \{  q_i \in \R^d \quad \hbox{ for } i =1,\ldots, n \}
 \]
 which satisfy  momentum conservation 
 \begin{equation}  \label{MC}
 \sum_{i=1}^n  q_i =0\ .
 \end{equation} 
 The external half-edge $i$ represents an incoming particle carrying  momentum $q_i$. 
 Given such a collection of momenta, we obtain a vector at every vertex
 \begin{equation} \label{MomentumateachVertex} 
  \underline{q} = \left(q_v\right)_v  \in   \left(\R^d\right)^{V(G)} \quad \hbox{ where  }\quad  q_v =\sum_{ i: \underline{n}(i)=v}  q_i \ , 
 \end{equation} 
which is defined to be $0$ if vertex $v$ has no external edges attached to it, and otherwise is the sum of all incoming momenta at $v$.
\subsection{A category of graphs} Let us fix integers $g, n, M\geq 0$. 

\begin{defn} A \emph{marked, weighted metric graph with masses and momenta} is  a tuple
$
(G, w, \underline{n}, \underline{M}) 
$
where $G$ is a connected graph, and 
\begin{eqnarray}  \label{ExtradataonGraphs}
w \ : \  V(G) & \To&  \N \cup \{0\} \\ 
\underline{n} \ : \  \{1,\ldots, n\} & \To & V(G) \nonumber \\ 
\underline{M} \ : \  E(G)  & \To & \{0,\ldots,M\} \nonumber 
 \end{eqnarray} 
 are the data of: a weighting $w$ for every vertex; $n$ external half lines $\underline{n}$; and an assignment $\underline{M}$ of a  mass label to each internal edge of $G$.  Its genus is 
 \[ 
g(G) =  h_G + \sum_{v\in V(G)} w(v) 
\]
where $h_G$ is the first Betti (loop) number of $G$. Such a graph is called \emph{stable} if, for every vertex $v\in V(G)$, one has the inequality 
$ 2 w(v) -2 + d(v)   \geq 0 $ 
 where $d(v)$ is the total degree at $v$ (including both internal and external half-edges). 
 \end{defn} 
 
 An isomorphism of graphs  $(G, w, \underline{n}, \underline{M})$  is an isomorphism of the underlying graphs which respects the additional data \eqref{ExtradataonGraphs}.
 
 For every edge $e\in E(G)$, the contraction 
 \[(G, w, \underline{n}, \underline{M} ) /e = (G/e, w', \underline{n}', \underline{M}' ) \ ,
 \]
 of $e$ is defined as follows. 
 If $e$ has distinct endpoints, then $G/e$ is the graph with edge $e$ removed and its endpoints $\{s(e),t(e)\}$ identified.  As usual, the weight of the new  vertex is the sum of weights $w(s(e))+ w(t(e))$, and $\underline{n}'$ is the composite  of $\underline{n}$  
 with the quotient map $V(G) \rightarrow V(G')= V(G)/ \{s(e) = t(e)\}$.  The map $\underline{M}' $ is the restriction of $\underline{M} $ to the subset $E(G/e) = E(G)\backslash \{e\}$. 
 
 When $e$ is a tadpole or self-edge, and its endpoints coincide $v=s(e)=t(e)$,   the contraction   $G/e$ is simply the deletion of edge $e$. The weight of the vertex  $v$  in $G/e$ is defined to be one more than its weight in $G$: $w'(v)= w(v) +1$. 
 The map $\underline{n}$ is unchanged, and $\underline{M}'$ is defined by restriction of $\underline{M}$ to $E(G) \backslash \{e\}$ as before. 
  
  Edge contraction preserves both the genus and the property of stability.

  \begin{defn}
   For fixed $g,n, M$, let 
  $I_{g,n,M}$ denote the category whose
  objects
  $   (G, w, \underline{n}, \underline{M})  $
     are connected stable graphs  of  genus $g$, and whose morphisms are generated  by   isomorphisms and edge contractions. 
  \end{defn} 
When $M=0$, it is nothing other than the category of marked stable metric graphs of genus $g$ with $n$ external markings. 
The category $I_{g,n,M}$ has a final object given by  the graph consisting of a single vertex of weight $g$, and $n$ external legs.

\subsection{Associated moduli space}
Based on ideas in  \cite{ACP,TopWeightAg},  one can efficiently define the moduli space of marked metric graphs as follows.

Let $\Top$ denote the category whose objects are topological spaces and whose morphisms are continuous maps. 
Consider  a functor $\cone: I \rightarrow \Top$, where $I$ is a finite  (diagram) category.  Its  realisation $X_{\cone}$ is the topological space
\[ X_{\cone} = \varinjlim_{i \in I} \cone(i)  \ .
\] 
A morphism from   $\cone: I \rightarrow \Top$   to
$\cone': J \rightarrow \Top$   is the data of a functor $j: I\rightarrow J$, together with a natural transformation from $\cone$ to $\cone' \circ j$.
Such a morphism induces a continuous map $X_{\cone} \rightarrow X_{\cone'}$  between the associated topological realisations. 

\begin{defn}  Define a functor 
\begin{equation}
\cone: I^{\mathrm{opp}}_{g,n,M} \To \Top\ 
\end{equation} 
as follows. To the object    $(G,w,\underline{n}, \underline{M})$  it associates the  closed cell 
\[
\overline{C}(G)= \overline{C}(G,w,\underline{n}, \underline{M}) = \R_{\geq 0}^{E(G)}, 
\]
 given by the space of non-negative edge lengths $\ell_e$, for $e\in E(G)$. It   only depends on the underlying graph $G$. 
 An  isomorphism 
$(G,w,\underline{n}, \underline{M}) \cong (G',w',\underline{n}', \underline{M}')$
induces a bijection  $E(G)\cong  E(G')$, and   a linear isomorphism on the corresponding cells. 
Contraction of an edge $e \in E(G)$ induces a natural  inclusion of cells  \[
\iota_e :  \overline{C}(G/e ) \hookrightarrow  \overline{C}(G) \]
upon identifying $ \overline{C}(G/e)$  with the  subset $\ell_e=0$.
\end{defn}

 \begin{defn}
Define the tropical moduli space to be its topological realisation
\[ 
\mathcal{M}^{\trop}_{g,n,M} =  \varinjlim_{G \in  I^{\mathrm{opp}}_{g,n,M}}  \overline{C}(G) \  .
\]
\end{defn}
Although the cases when $g=0$ are topologically interesting, they will play no role in this paper: in the absence of internal edges the canonical differential forms we shall consider are identically zero and only interesting for $g>0$. 

Similarly, one can consider the functor $L \cone: I^{\mathrm{opp}}_{g,n,M} \rightarrow \Top\ $ which assigns to every object
  $(G,w,\underline{n}, \underline{M})$  the closed simplex
\[ 
\overline{\sigma}_G =   \{ \left(\ell_e\right)_e \in \R_{\geq 0}^{E(G)}   \ : \ \sum_{e\in E(G)} \ell_e =1 \}  
\]
where the edge lengths are normalised to sum to $1$.  If $E(G)$ is empty, then $\overline{\sigma}_G$ is defined to be the empty set.  

 It defines a functor for the same reasons as above.  The tropical moduli space $\mathcal{M}^{\trop}_{g,n,M}$ is a cone whose cone point $p$  is the cell associated to the final object in $I_{g,n,M}$.   The  link $\mathcal{M}^{\trop}_{g,n,M} \backslash \{p\}/\R^{\times}_{>0} $ is isomorphic to the topological space
\[ 
L\mathcal{M}^{\trop}_{g,n,M} =  \varinjlim_{G \in  I^{\mathrm{opp}}_{g,n,M}}  \overline{\sigma}_G \ .
\]

\begin{rem} In the case when $M=0$, one retrieves the definition of the moduli space $ \mathcal{M}^{\trop}_{g,n}$ of curves with $n$ marked points.  In general, the data of masses $\underline{M}$ is not completely anodine, since assigning distinct masses to the edges of a graph will in general reduce the size of its automorphism group. 
\end{rem}
The functor $\cone$ defined above factors through a category of rational polyhedral cones as in \cite{ACP,TopWeightAg}, which can in turn be upgraded to a category built out of affine spaces equipped with certain linear embeddings. There is a variant for the links  $L\cone$  involving projective spaces.   
These notions  will be pursued elsewhere.

\subsection{Homology and graph complexes}
The homology of the spaces $ \mathcal{M}^{\trop}_{g,n,M}$ can be expressed in terms of complexes of graphs in the category $I_{g,n,M}$ via a 
variant of the classical result which relates cellular and singular homology   (see \cite[Proposition 2.1]{TopWeightAg}, \cite[Theorem 4.2]{ACP}).  
We refer the reader to   \cite{ModuliColoured} for specific results about the homology of  $ \mathcal{M}^{\trop}_{g,n,M}$ when $M>0$. 

\section{Interpretations of momenta}  \label{sect: MomInterp}
 
\subsection{Genericity}
Let $G$ be a connected graph with external momenta $q_i \in \R^d$, and internal masses $m_e$. We shall say that momenta are \emph{generic} if  
\begin{equation}  \label{Genericity} 
\left(\sum_{i\in I} q_i \right)^2 \neq 0 
\end{equation} 
for all  strict subsets $I $ of the set of external half edges  \cite[(1.18)]{Cosmic}. It is automatic from our definitions that  all non-trivial edge masses $m_e$ are non-zero.  Note that the genericity conditions for complex-valued momenta, which will  not be considered  in this paper, are more general than the Euclidean condition above.

\subsection{Momentum routing} 
Choose an orientation   on every  edge of $G$, such that  external edges are oriented inwards.

\begin{defn}  
A  \emph{momentum routing}  relative to the edge orientation  is the data, for every internal  edge  $e\in E(G)$,  
  of a  momentum vector
\[
\mu_e \in \R^d  
\]
such that momentum conservation holds at every
 vertex $v\in V(G)$, i.e., 
\begin{equation} \label{localMC} 
\sum_{s(e) = v} \mu_e- \sum_{t(e) = v}  \mu_e  = 0 \ ,
\end{equation} 
where the first sum is over all edges emanating from $v$, and the second is over all edges terminating at $v$. In the second sum we include external half edges, for which we write $\mu_e= q_e$ for the incoming  external momentum along $e$. 
\end{defn}
To see that momentum routings exist, we can reformulate \eqref{localMC} as follows. Let us write $R = \R^d$, and consider the exact sequence:
\[ 
0 \To H_1(G;R) \To R^{E_G} \overset{\partial}{\To} R^{V_G} \To R \To 0 \ ,
\]
where $\partial(e) = t(e) - s(e)$. 
The data of external momenta defines a vector  $ \underline{q} \in R^{V_G}$ where $ \underline{q}= (q_v)_{v\in V_G}$   via  \eqref{MomentumateachVertex}.  The condition of momentum conservation \eqref{MC} is precisely the statement that its image under the map $R^{V_G} \rightarrow R$ is zero. Therefore $\underline{q} \in \mathrm{Im} (\partial)$. A choice of momentum routing is any element 
\[ 
\underline{\mu}= (\mu_e)_{e\in E_G} \in  R^{E_G}  \quad \hbox{ such that } \quad \partial  \underline{\mu} = \underline{q}  \ .
\] 
Any two  choices of momentum routing differ by an element of $H_1(G;R)$.

Note that one may always assume that $\mu_e=0$ for every self-edge or tadpole $e$. 

\begin{rem} 
One can define a canonical momentum routing by demanding that $\underline{\mu} \in R^{E_G}$ be orthogonal to the subspace $H_1(G;R)$ with respect to the inner product on $R^{E_G}$ for which the edges form an orthonormal basis. However, the canonical routing is not respected by contraction of edges. 
\end{rem} 
\subsection{Complex momenta}  Suppose that the number of spacetime dimensions is $d=2$, and  the external momentum associated with the $\mathrm{k}^{th}$ external half edge is 
\begin{equation} \label{qinR2}
q_k \in \R^2\ .
\end{equation} 
Having fixed a  choice of square root $i$ of $-1$,   there is an identification
\begin{eqnarray}
 \R^2  & \cong &  \C \nonumber \\ 
 q=(x,y) & \mapsto & \cq= x+iy \nonumber 
\end{eqnarray} 
which is compatible with the Euclidean norm. In other words, if $q_k = (x_k,  y_k)$ and $\cq_k = x_k+iy_k$ for $k=1,2$, then  the Euclidean inner product is
\begin{equation} \label{EuclideanToComplex} 
2\,  q_1. q_2  =   2 \left(x_1 x_2 + y_1 y_2\right) = 
 2 \, \mathrm{Re} (\cq_1 \overline{\cq}_2  )  =  \cq_1 \overline{\cq}_2 + \overline{\cq}_1 \cq_2  
 \end{equation}
and   $|\!|q_1|\!|^2 = \cq_1 \overline{\cq}_1$. 
We can always assume \eqref{qinR2} holds when the total number of external momenta is at most $3$, since by momentum conservation, the Feynman integral only depends on two external momenta, which  lie in a Euclidean  plane.

\subsection{Quaternionic momenta} 
\label{sect: QuatMomenta} Suppose that $G$ is a connected graph as above, but the number of Euclidean dimensions is $d=4$. The external momenta satisfy 
\begin{equation} \label{qinR4}
q_k = (q_k^{(1)}, \ldots, q_k^{(4)}) \in  \R^4\ .
\end{equation}
Let us denote the corresponding quaternions   \S\ref{sect: Quaternions} by  $\cq_k  \in \Qt $, where  
\[
 \cq_k = q_k^{(1)}\Qu + q_k^{(2)}\Qi +   q_k^{(3)}\Qj 
+ q_k^{(4)}\Qk \ , 
\]
which we may identify with their image in $M_2(\C)$ under $\chi$ \eqref{chidef}. The inner product will be written interchangably in either notation via the identity:
\begin{equation} 
  \cq_i \overline{\cq}_j + \cq_j \overline{\cq}_i   =  (2  q_i . q_j )  \Qu =  2 \sum_{k=1}^4  q_i^{(k)} q_j^{(k)}\ .
\end{equation} 

\section{Generalised Laplacians for graphs with external legs} \label{sec: GenLaplace}

We define a  generalised graph Laplacian of a  connected graph $G$ 
with external momenta subject to momentum conservation,  and internal particle masses.
In the following, let  $\CC$ be a ring, which is not necessarily commutative, equipped with an anti-commutative involution $\iota: \CC \rightarrow \CC$ such  that its invariant subspace $\mathcal{R} = \{ c\in \CC: \iota(c) = c\}$ is commutative. 

Examples include:
\begin{enumerate}
\item $\CC= \C$ the complex numbers, equipped with complex conjugation.
\vspace{0.05in} 

\item  $\CC = \Qt$ the ring of quaternions, with quaternionic conjugation. 
\vspace{0.05in} 

\item $\CC= \Q[ \mu_e, \overline{\mu}_e,  m_i] $ an abstract ring of  kinematic variables, where $\iota$ is the $\Q$ linear map acting trivially on $m_i$, and such that $\iota (\mu_e) = \overline{\mu}_e$, $\iota(\overline{\mu}_e) = \mu_e$. 
\vspace{0.05in} 

\item $\CC= \Q[ \lambda_e, \overline{\lambda}_e,  \nu_e, \overline{\nu}_e, m_i] $ as in $(3)$, 
where $\iota( \lambda_e) = \overline{\lambda}_e$,  and  $\iota( \nu_e) = -\nu_e$, $\iota( \overline{\nu}_e) = -\overline{\nu}_e$. The multiplication is given by viewing $\mu_e = \lambda_e + \Qj \nu_e$ as an abstract quaternion, or equivalently as  a $2 \times2$ matrix $\left( \begin{smallmatrix} 
\lambda_e & \nu_e \\ -\overline{\nu}_e &  \overline{\lambda}_e \end{smallmatrix} \right)$. 
 \vspace{0.05in}

\end{enumerate}
In the first instance, the reader may wish to consider only the  case $(1)$ where  $\CC=\C$. 
The generalised graph Laplacian requires the data of: 
\begin{itemize}

\item  an external momentum vector $\underline{q} \in \CC^{V_G}$ subject to momentum conservation  

\item a vector of internal masses $\underline{m}\in \CC^{E_G}$ which is invariant under $\iota$.  

\end{itemize} 
It is defined in two stages: first we define the graph Laplacian for zero masses, and then modify it very slightly to take into account the internal masses. 

\subsection{Definition of the generalised graph Laplacian} \label{sect: DefGraphLaplacianComplex} 
Choose an orientation on the internal edges of $G$. 
Consider the exact sequence 
\[
0 \To H_1(G; \CC)\overset{\mathcal{H}_G}{\To} \CC^{E_G}  \overset{\partial}\To \CC^{V_G} \To \CC \To 0  \ . 
\] 
Since the momentum vector $\underline{q} \in \CC^{V_G}$ maps to zero, we obtain an  extension 
\begin{equation} \label{ExtensionofH1G} 
0 \To H_1(G;\CC)  \To \mathcal{E}_G \To \CC \To 0
\end{equation} 
where $\mathcal{E}_G =\{ \underline{\mu}  \in \CC^{E_G} :  \partial( \underline{\mu}) \in \CC \underline{q}  \}$ and the element $1\in \CC$ in the right-hand term is identified with $\underline{q} \in \CC^{V_G}$. We shall call this the \emph{extension by momenta}. 

  Consider the pairing  satisfying
\begin{eqnarray} \label{HermformonCEG}
\CC^{E_G} \times \CC^{E_G} &  \To &  \CC[x_e, e\in E_G]   \\ 
\langle \alpha e , \beta e'\rangle_0  & = &  \alpha\,  \iota(\beta) \, x_e \delta_{e,e'} \ . \nonumber 
\end{eqnarray}
It restricts to a Hermitian form (with respect to the involution $\iota$) on $\mathcal{E}_G$, which we also denote by $\langle \ , \rangle_0$, and which may be interpreted as a  $\CC$-linear map  $ y\mapsto \left( x\mapsto \langle y, x  \rangle_0\right) $ which we call the `graph Laplacian for zero masses':
 \[
\widetilde{\Lambda}_G(\mu,0):  \mathcal{E}_G \rightarrow  \mathrm{Hom}_{(\CC,\iota)}
(\mathcal{E}_G,  \CC[x_e] ) \ .
 \] 
 Here, $\mathrm{Hom}_{(\CC,\iota)}(A,B)$ denotes anti-linear maps $\phi:A \rightarrow B$ satisfying   $\phi(\lambda x) = \iota(\lambda) \phi(x)$ for all $\lambda \in \CC$. 
Equivalently,  \eqref{HermformonCEG} may be viewed as  a $\CC$-linear map:
  \[
  D_G : \CC^{E_G} \To  \mathrm{Hom}_{(\CC,\iota)}(\CC^{E_G},  \CC[x_e] )
  \] 
  using which   one has the identity
$
\widetilde{\Lambda}_G(\mu,0)= \widetilde{\mathcal{H}}_G^* D_G \widetilde{\mathcal{H}}_G    $
 where  $\widetilde{\mathcal{H}}_G: \mathcal{E}_G \rightarrow \CC^{E_G}$ is inclusion and $X\mapsto X^*$ denotes Hermitian conjugation, i.e., $X^* = \iota(X)^T$.

The subspace $H_1(G;\CC) \subset \mathcal{E}_G$ contains the subspace  $H_1(G;\Z)$, which is invariant under the involution $\iota$ .
Restricting the Hermitian form $\langle \ , \  \rangle_0$    to the space $H_1(G;\Z)$ therefore  defines a symmetric bilinear form which is nothing other than the usual graph Laplacian $\Lambda_G$, which depends neither on masses nor momenta. 

The masses of $G$  are encoded by a single element
\[ 
\underline{m}= \sum_{e} m_e e   \  \in \  \CC^{E_G}
\]
where $m_e = \iota(m_e)$. 
The norm of this element with respect to \eqref{HermformonCEG}  is 
\[
|\!| \underline{m} |\!|^2 =  \langle  \sum_{e} m_e e ,  \sum_e m_e  e \rangle = \sum_e m_e^2 x_e \]
We define the \emph{mass-correction} Hermitian inner product $\langle   \ , \ \rangle_m$  on $ \mathcal{E}_G$ to be  the unique inner product vanishing identically on $H_1(G;\CC)$:
 \[
 \langle H_1(G;\CC)  \ ,  \   \mathcal{E}_G \rangle_m =  \langle    \mathcal{E}_G  \ , \ H_1(G;\CC)  \rangle_m=0 \]
  but for which 
 $ \langle f, f \rangle_m =  |\!| \underline{m} |\!|^2$ for any $f\in \mathcal{E}_G$ whose image is $1\in \CC$ under the natural map in  \eqref{ExtensionofH1G}.  Consider the sum of the two Hermitian forms
 \[ 
 \langle  \ , \ \rangle =  \langle  \ , \ \rangle_0 +  \langle  \ , \ \rangle_m  \ .
 \]
The associated $\CC$-linear map defines the generalised graph Laplacian 
\begin{equation}  \label{def: LambdaGqm} 
\widetilde{\Lambda}_G(\mu,m) : \mathcal{E}_G \To  \mathrm{Hom}_{(\CC, \iota)} (\mathcal{E}_G,  \CC[x_e] )  \ . 
\end{equation}

\subsection{Generalised graph Laplacian matrix}
The generalised graph Laplacian may be computed  by splitting \eqref{ExtensionofH1G}. Number the edges of $G$ from $1,\ldots, N$.  Choose a basis $c_1,\ldots, c_h$ of $H_1(G;\Z)$ and a routing of edge momenta $\underline{\mu} \in \CC^{E_G}$.

Define a  matrix $\widetilde{\mathcal{H}}_G$ with  $N$ rows, $h+1$ columns and entries in $\CC$  as   follows.
If $1 \leq k \leq h$ then the  entries $(\widetilde{\mathcal{H}}_G)_{e,k}$ count the number of times (with multiplicity) that the oriented edge $e$ appears in the cycle $c_k$. Thus the first $h$ columns consists of  the usual \emph{edge-cycle incidence matrix} $\mathcal{H}_{G}$. 
 The final column is defined by 
\[
\left(\widetilde{\mathcal{H}}_G\right)_{e,h+1} = \mu_e \qquad 1\leq e \leq N\ .
\]
Now consider  the matrix $\widetilde{M}_G$ with $N$ rows and $h+1$ columns  
\[ 
(\widetilde{M}_G)_{e,c} = \begin{cases}  m_e  \quad \hbox{ if  } \quad c= h+1 \ , \\ 0  \quad \hbox{ otherwise } . \end{cases}\ .
\]
Thus $\widetilde{M}_G$ is zero except for the last column, which is the vector of masses $m_e$.

 The generalised graph Laplacian matrix is defined to be 
 \[
 \widetilde{\Lambda}_G(q,m) =   \widetilde{\mathcal{H}}_G^{*}  D_G \widetilde{\mathcal{H}}_G +    \widetilde{M}_G^*  D_G \widetilde{M}_G\ , 
 \]
 where $X\mapsto X^*$ denotes  Hermitian conjugation, i.e., $X^* = \iota(X)^T$. Since $\widetilde{M}_G$ has entries which are invariant under $\iota$,  we have  $\widetilde{M}_G^{*} = \widetilde{M}_G^T$. Note that the matrix 
  $ \widetilde{M}_G^*  D_G \widetilde{M}_G$  only has a single non-zero entry in the bottom right-hand corner:
 \[
  \widetilde{M}_G^*  D_G \widetilde{M}_G  = \sum_{e=1}^N  \begin{pmatrix}  0  & \ldots  & 0 & 0  \\
\vdots  &  \ddots & \vdots  & \vdots  \\
  0  &    \ldots & 0 & 0  \\
     0    & \ldots &  0  &   m_e^2 \alpha_e  \\
      \end{pmatrix} \ .
      \]
  In general, the matrix $\widetilde{\Lambda}_G$  has the following block-matrix form:
  \begin{equation}\label{LambdaGblockmatrix} 
    \widetilde{\Lambda}_G = \left(
  \begin{array} {ccc|c}
  &  & &  c_1(\alpha ,\mu ) \\
  &  \Lambda_G & &  \vdots \\
   &  & &  c_h(\alpha ,\mu ) \\
   \hline
 c_1(\alpha , \iota(\mu) )    & \ldots  & c_h(\alpha , \iota(\mu)) & X_G 
  \end{array}  \right)  \ ,  \end{equation}
  where $\Lambda_G = \mathcal{H}_G^T D \mathcal{H}_G $ is the usual graph Laplacian matrix, 
  \begin{equation} \label{Xdef} 
   X_G = \sum_{e=1}^N  (\mu_e \iota({\mu}_e)  + m_e^2) \alpha_e  \ ,  
  \end{equation} 
and where for any $\Z$-linear combination of edges $c = \sum_{e\in E(G)} p_e e$, we write
\[
c(\alpha, \mu) = \sum_{e \in E(G)} p_e \mu_e \alpha_e \ ,
\]
and $c(\alpha , \iota(\mu))$ for its image under $\iota$.

\subsubsection{Change of edge orientations}
Changing the orientation of an edge does not modify the matrix $\widetilde{\Lambda}_G$.  To see this, observe that 
switching the orientation of any number of edges amounts to multiplying the matrix $\widetilde{\mathcal{H}}_G$ on the left by a diagonal matrix $E$ with entries in $\{1, -1\}$: reversing the orientation of edge $e$ results in a change of sign for the associated momentum $\mu_e$.  The claim follows since $E^* D_GE=D_G$ and since $\widetilde{M}_G$  does not depend on the edge orientations. 

\subsubsection{Change of basis} \label{sect: ChangeofBasis} The matrix $\widetilde{\Lambda}_G$  depends on the choice of basis for the homology $H_1(G;\Z)$.  Changing basis modifies the matrix $\widetilde{\Lambda}_G$ by 
\begin{equation} 
\label{ChangeLambdabyP}
 \widetilde{\Lambda}'_G =  \widetilde{P}^T \widetilde{\Lambda}_G \widetilde{P}
 \end{equation} 
where $\widetilde{P}$ is an invertible block matrix of the form 
\begin{equation} \label{formofP} 
\widetilde{P} = 
\left(
  \begin{array} {ccc|c}
  &  &   &  0\\
  &  P  & &  \vdots \\
    &  & &   0 \\
   \hline
0     & \ldots  & 0  & 1 
  \end{array}  \right)  \ , 
\end{equation}
where  $P \in \GL(H_1(G;\Z)) $ is a change of basis matrix  with integer entries.   In particular,  one has $\det(P)^2 =\det(\widetilde{P})^2=1$.

\subsubsection{Change of momentum routing} 
Two momentum routings $\overline{\mu}$, $\overline{\mu}'$ differ by an element of $H_1(G;\CC)$. Therefore changing momentum routing (or changing 
a choice of splitting of \eqref{ExtensionofH1G}) is equivalent to modifying the matrix $\widetilde{\Lambda}_G$ by
\begin{equation} 
\label{ChangeLambdabyS}
 \widetilde{\Lambda}'_G =  \widetilde{S}^* \widetilde{\Lambda}_G \widetilde{S}
 \end{equation} 
where $\widetilde{S}$ is an invertible block matrix of the form 
\begin{equation} \label{formofP} 
\widetilde{S} = 
\left(
  \begin{array} {ccc|c}
  &  &   &  s_1\\
  &  I  & &  \vdots \\
    &  & &   s_h \\
   \hline
 0 & \ldots  & 0  & 1 
  \end{array}  \right)  \ , 
\end{equation}
where $I$ denotes the $h_G \times h_G$ identity matrix and $s_1,\ldots, s_h\in \CC$. 

\subsubsection{Group of indeterminacy}
The  graph Laplacian matrix $\widetilde{\Lambda}_G$ is therefore ambiguous up to the action of the semi-direct product
\begin{equation} \label{GroupofIndeterminancy} 
 \CC^{g}    \rtimes\GL_g(\Z) 
\end{equation}
where $P \in \GL_g(\Z)$ acts on  $S\in \CC^g$ via  $S\mapsto PSP^{-1} $.    Note that it only involves the additive structure on $\CC$, which is commutative, and not the multiplicative structure, which may not be.  The group \eqref{GroupofIndeterminancy} can be viewed as the group of automorphisms of the extension 
$\mathcal{E}_G$  \eqref{ExtensionofH1G} which respects the integral structure $H_1(G;\Z) \subset H_1(G;\CC)$. 
Recall that $   \CC^{g}    \rtimes\GL_g(\Z) 
 $ is the set 
$(S, P) \in \CC^g \times \GL_g(\Z) ,$ which can be identified with the set of matrices
\begin{equation} 
SP = 
\left(
  \begin{array} {ccc|c}
  &  &   &  s_1\\
  &  P  & &  \vdots \\
    &  & &   s_h \\
   \hline
 0 & \ldots  & 0  & 1 
  \end{array}  \right)  \ , 
\end{equation}
equipped with the group law $(S_1,P_1) (S_2,P_2) = (S_1 P_1 S_2 P_1^{-1}, P_1P_2)$. 

\subsection{Complex adjoint of the generalised graph Laplacian}
Let $G$ be a connected graph as before, and suppose that the external momenta lie in $\CC = \Qt$, the ring of quaternions. 
The generalised graph Laplacian 
$
\widetilde{\Lambda}_G(\mu, m) 
$
is a quaternionic Hermitian form of rank $g+1$, and may be represented by a $(g+1) \times (g+1)$ Hermitian matrix with entries in $\Qt$. Its image under the complex adjoint map 
\eqref{chidef} is a  Hermitian complex matrix of rank $2g+2$.

In block matrix form \S\ref{sect: QuatMatrices} \emph{(i)}, one may represent a  choice of complex adjoint graph Laplacian 
in the form $M_{2n}(\C) \cong M_n(M_2(\C))$ as follows:
\begin{equation} \label{ExpandedLaplacianBlockMatrix} 
\chi_{ \widetilde{\Lambda}_G(\mu, m) } =  
\left(
\begin{array}{cccc|cc}
   &   &    &&   c_1(\lambda,\alpha)  &  c_1(\nu, \alpha)   \\
  &   \chi_{\Lambda_G}&    && - c_1(\overline{\nu},\alpha)  &  c_1(\overline{\lambda},\alpha)      \\
  &  &  & &&   \vdots \\ 
  &  &  &   &&  \\ \hline  
   c_1(\overline{\lambda},\alpha) &  - c_1(\nu,\alpha)    & \ldots && X & 0  \\
   c_1(\overline{\nu},\alpha) &   c_1(\lambda,\alpha)   & \ldots &  &0 &  X 
\end{array}
\right)
\end{equation} 
where  the quaternionic momenta are $\mu_e = \lambda_e +  \Qj\,  \nu_e $, for $\lambda_e, \nu_e \in \C$,  where
\begin{equation} \label{expandedLambdaG} 
\chi_{\Lambda_G}=  
\left(
\begin{array}{ccccc}
  \Lambda_{1,1}&  0    &  \Lambda_{1,2} & 0 & \ldots \\
  0 &     \Lambda_{1,1} & 0  &    \Lambda_{1,2}  &    \\
   \Lambda_{2,1} & 0 &  & &  \\ 
 0 &   \Lambda_{2,1} & \ddots  &  &   \vdots \\
\vdots  &   &      &  \Lambda_{g,g} & 0     \\
  & &  \ldots & 0 &   \Lambda_{g,g} \\
\end{array}
\right)
\end{equation} 
is the complex adjoint  of the ordinary graph Laplacian $\Lambda_G= (\Lambda_{ij})_{ij}$, and 
\[ 
X= \sum_{e=1}^{e_G} \left( m_e^2 + \lambda_e\overline{\lambda}_e + \nu_e \overline{\nu}_e \right)  \alpha_e  \ .
\]
In the  massless case $m_e=0$, these matrices  were considered in \cite{BlochKreimer}.

The reader may prefer to use the  equivalent matrix  representation  \S\ref{sect: QuatMatrices} \emph{(ii)}.

\section{Symanzik and spanning forest polynomials} \label{sect: Symanzik}
We  recall the definition of Symanzik polynomials and  relate the  polynomial $\phi_G(q)$ to a
sum of spanning forest polynomials relative to a  choice of momentum routing.  
For further background on Symanzik polynomials and some of their basic properties, we refer the reader to  \cite[\S1]{Cosmic}, for example. 
\subsection{Definition of Symanzik polynomials}  \label{sect: DefSymanzik} 
We recall the definition of the graph polynomials which arise in  Feynman integrals in parametric form, since conventions can vary slightly in the literature.  Let $d\geq 0$ be any non-negative integer. 
In this section,  $G$ is any connected graph with external edges and external momenta 
$
q_i \in \R^d\ .$
 Every internal edge $e$ is assigned a mass $m_e \in \R$. 
  
\subsubsection{1st Symanzik polynomial}  
Recall that the first Symanzik polynomial (also known as the `Kirchhoff' or graph polynomial), is defined by 
\begin{equation}  \label{Psidef} 
\Psi_G = \sum_{T\subset G}  \prod_{e\notin T} \alpha_e
\end{equation} 
where the sum is over all spanning trees of $G$. It does not depend on the external half-edges of $G$. 
It is equal to the determinant of the ordinary graph Laplacian:
\begin{equation} \det \Lambda_G = \Psi_G\  . \end{equation}
It is homogeneous of degree $h_G$ and is not identically zero. 

\subsubsection{2nd Symanzik polynomial}  The second Symanzik polynomial  depends on external momenta. It is the homogeneous polynomial  of degree $h_G+1$
 defined by 
\begin{equation} \label{Phidef}
\Phi_G(q) =  - \sum_{T_1,T_2\subset G}    (q^{T_1}.  q^{T_2}) \prod_{e \notin T_1 \cup T_2} \alpha_e
\end{equation} 
where the sum is over all spanning 2-trees of $G$ (or spanning forests with exactly two connected components), and $q_{T_i}$ denotes the total momentum entering $T_i$. By momentum conservation $q^{T_1} + q^{T_2}=0$.  

\subsubsection{Dependence on masses} The  `2nd Symanzik' polynomial which occurs in the  (Euclidean) parametric representation of Feynman integrals is 
\begin{equation} \label{Xidef} 
\Xi_G(q, m) = \Phi_G(q) + \left(\sum_{e \in E_G} \alpha_e m_e^2 \right) \Psi_G
\end{equation} 
which is also homogenous of degree $h_G+1$.

\subsection{Vanishing of Symanzik polynomials}
\begin{defn} \cite[\S1.4]{Cosmic}
A subgraph $\gamma \subset E_G$ is \emph{momentum-spanning} if  all  non-zero  external momenta   $q_v$ of $G$ meet   vertices $v$ which lie  in  a single connected component of $\gamma$. It is called \emph{mass-spanning} if it contains all massive edges of $G$.

It is called \emph{mass-momentum spanning} (or \emph{m.m.} for short) if both  hold.  \end{defn} 

It is shown in \cite[\S1.6]{Cosmic} that if $\gamma \subset E_G$ and genericity \eqref{Genericity} holds,  then 
\begin{equation} 
\Xi_{G}(q,m) \Big|_{\alpha_e =0, e\in \gamma} \   =  \  \Xi_{G/\gamma}(q,m)  
\end{equation}  
is identically zero  if and only if $\gamma$ is an \emph{m.m.} subgraph. 
In this paper we work in  the Euclidean region (\emph{loc. cit. } \S1.7) which implies that  the coefficients of every monomial in the variables  $\alpha_i$ in   $\Xi_G$ are  non-negative (this would not be true in Minkowski space). If $\gamma$ is not \emph{m.m.} then  $G/\gamma$ has a scale and one shows that
\[  \Xi_{G/\gamma }(q,m)   >0    \]
in the region $\alpha_e>0$, for all $e\in E_{G/\gamma}=E_G \backslash E_{\gamma}$. 

\subsection{Spanning forest and Dodgson polynomials}
The second Symanzik  can be expressed using  spanning forest polynomials, which were defined in \cite{BrownYeats}.

\subsubsection{Spanning forest polynomials} 
\begin{defn} Let $P= P_1 \cup \ldots \cup P_r$ be any partition of a subset of vertices of $G$ into disjoint sets. The associated spanning forest polynomial is defined by 
\[
\phi_G^P = \sum_{\mathcal{F} } \prod_{e \in E(G) \backslash \mathcal{F}} \alpha_e
\]
where the sum is over spanning forests $\mathcal{F} = T_1 \cup \ldots \cup T_r$ of $G$ with exactly $r$ connected components $T_i$  such that 
\[
V(T_i) \cap P = P_i \qquad \hbox{ for } \quad 1\leq i \leq r \ .
\] 
In other words, each tree $T_i$ contains the vertices in $P_i$ and no other vertices of $P$. Trees consisting of a single
vertex are permitted.  
We  set $\phi_G^P=0$ if the sets $P_i$ are not all disjoint. 
\end{defn}

 \subsubsection{An identity for the 2nd Symanzik polynomial}
 \begin{prop}  \label{prop: PhiGasForest}  Let $G$ be as in the previous paragraph. For any choice of momentum routing, the  second Symanzik polynomial satisfies:
\begin{multline} \label{PhiGasForest} 
 \Phi_G(q) = \sum_{e\in E(G)}  \mu_e . \mu_e \, \alpha_e \Psi_{G\q e}  \\
 +  \sum_{\substack{e,f \in E(G) \\ e\neq f }} \mu_e. \mu_f \,  \alpha_e\alpha_f \left( \phi_G^{\{s(e),s(f)\} , \{t(e),t(f)\} } -  \phi_G^{\{s(e),t(f)\} , \{s(f),t(e)\} }\right)\ , 
\end{multline}
where we recall that  $\phi_G^{P_1, P_2}$ is  zero  if $P_1 \cap P_2 $ is non-empty.  Note that the second summand is symmetric in $e$ and $f$ and therefore each term is counted twice. 
\end{prop} 
\begin{proof}
Let $T= T_1 \cup T_2$ be a spanning forest in $G$ with two connected components $T_1,T_2$.  
By momentum conservation \eqref{localMC}, the momentum flowing into $T_1$ equals
\begin{equation} \label{qT1formula}
q^{T_1} = - q^{T_2} = \sum_{t(e) \in T_1}  \mu_e - \sum_{s(e) \in T_1} \mu_e
\end{equation} 
where each sum is over all internal edges $e\in E(G)$ such that $e\cup T$ has no loops, and is hence a spanning tree of $G$. In particular, $e$ has one endpoint in $V(T_1)$ and the other in $V(T_2)$. 

Suppose that $e$ is an edge with distinct endpoints. 
It follows  from \eqref{qT1formula} that the coefficient of $\mu_e. \mu_e$ in $-q^{T_1}. q^{T_2}$ is $1$ if $e \cup T_1 \cup T_2$ is a tree, and $0$ otherwise. 
 By    \eqref{Phidef},  the coefficient of $\mu_e. \mu_e$ in  $\Phi_G(q)$  is therefore equal to
\[
\sum_{T }  \prod_{e' \notin T}  \alpha_{e'}
\]
where the sum is over the set of  spanning forests $T$ with two connected components such that $T \cup e$ is a tree. 
The map $T \mapsto T \cup e$ is a bijection between the latter
and the set of spanning trees which contain $e$, which in turn are  in one-to-one correspondence with the set of spanning trees in $G/e$. We conclude that the coefficient of $\mu_e. \mu_e$ in the formula \eqref{Phidef} is precisely $\alpha_e \Psi_{G/e}$ (which also equals  $ \alpha_e \Psi_{G\q e}$).

In the case when $e$ is a loop, $\Psi_{G\q e}$ vanishes. By \eqref{qT1formula}, $\mu_e$ does not contribute to $q^{T_1}$ and therefore does not appear in $\Phi_G(q)$.  This establishes the first line of \eqref{PhiGasForest}.

Now consider two distinct edges $e, f$ and assume first of all that they have  in total four  distinct endpoints. 
Formula \eqref{qT1formula} implies that $\mu_e. \mu_f$ does not arise in $-q^{T_1}.q^{T_2}$
unless both $T \cup e $ and $T\cup f$ are trees. In this situation the coefficient is $1$ if 
  $s(e), s(f)$ lie in the same connected  component  of $T$, and $-1$ if they lie in different components. In the former case, this means  precisely  that $T=T_1 \cup T_2$ is a spanning forest with 2 components such that  the vertices of one component $V(T_1)$ contains $\{s(e), s(f)\}$ and the other, $V(T_2)$, contains $\{t(e), t(f)\}$.  In the latter case, $V(T_1)$ contains  $\{s(e), t(f)\}$ and $V(T_2)$ contains $ \{s(f), t(e)\}$, or vice-versa.  From   \eqref{Phidef},  it follows that the coefficient of $\mu_e. \mu_f$ in $\Phi_G(q)$ is exactly 
  \begin{equation} \label{alphaefPhiGterm} 
  2 \, \alpha_e \alpha_f   \left( \phi_G^{\{s(e),s(f)\} , \{t(e),t(f)\} } -  \phi_G^{\{s(e),t(f)\} , \{s(f),t(e)\} }\right)\ .
  \end{equation} 
  It remains to check that \eqref{alphaefPhiGterm} remains true in the cases when $\{s(e),s(f), t(e), t(f)\}$ has 3 or fewer elements.
  The reader may verify this in every situation (e.g., $s(e) = s(f)$, $s(e)= t(f)$, etc), using a similar argument  and  the fact that $\phi_G^{P_1,P_2}$ vanishes if $P_1$ and $P_2$ are not disjoint.   
   Note that if $e$ or $f$ is a loop (i.e., $s(e)=t(e)$) then $\mu_e$ does not occur in \eqref{qT1formula}, and  indeed, both  of the  spanning forest polynomials in \eqref{alphaefPhiGterm} are defined to be zero in this case.     \end{proof}

\subsubsection{Spanning forests and minors of  graph Laplacians}
The difference of spanning forest polynomials occurring in proposition \ref{prop: PhiGasForest} can be interpreted as a `Dodgson polynomial' $\Psi_G^{e,f}$.  
We will make frequent use of the  following well-known lemma. 

\begin{lem}    \label{lem: choosecycles} Let $G$ be a connected graph with directed edges. Suppose that $e_1,\ldots, e_n $ are distinct edges such that 
$G\backslash \{e_1,\ldots, e_n\}$
is connected. Then $h=h_G \geq n$ and we may choose cycles $c_1,\ldots, c_h \in \Z^{E_G}$ representing a basis for $H_1(G;\Z)$ such that:
\[ 
\left( \mathcal{H}_G \right)_{e_i,c_j}   = \delta_{i,j}  \quad \hbox{for all } 1 \leq i\leq n  , 1\leq j \leq h_G \ .
\]
In other words, for each $1\leq i\leq n$, the edge $e_i$ appears in cycle $c_i$ with coefficient $1$ and occurs in no other cycle.
\end{lem}

For a matrix $M$, let $M^{i,j}$ denote the matrix with row $i$ and column $j$ removed. 

\begin{lem}   \label{lem: MinorsofLambdaG} Let $G$ be a connected graph with directed edges.

(i) Let $e\in E(G)$ be an edge such that $G\backslash \{e\}$ is connected (i.e., $e$ is not a bridge). Choose cycles $c_1,\ldots, c_h$ as in the previous lemma and let $\Lambda_G$ be the graph Laplacian matrix with respect to this basis.  Then we have 
\[
 \det \Lambda_G^{1,1}   =   \Psi_{G\backslash e}\ .
 \]
 
(ii).  Let $e_1,e_2\in E(G)$ be two edges such  that $ G\backslash \{e_1,e_2\}$ is connected, and choose cycles $c_1,\ldots, c_h$  as in the previous lemma. Then 
\[ 
\det \Lambda_G^{1,2}   =   \pm \left( \phi_G^{\{s_1,s_2\} , \{t_1,t_2\} } -  \phi_G^{\{s_1,t_2\} , \{s_2,t_1\}}   \right)   \ . 
\]
where  $s_i$ is the source of the directed edge $e_i$, and $t_i$ its target, for $i=1,2.$ 
\end{lem} 
\begin{proof} $(i)$. The variable $\alpha_e$ only occurs in the matrix $\Lambda_G$ in  a single place, namely the top left entry $(\Lambda_G)_{1,1}$. A row expansion shows that the coefficient of $\alpha_1$ in $\det (\Lambda_G)$ equals $\det (\Lambda_G^{1,1})$.  On the other hand, the coefficient of $\alpha_1$ in $\Psi_G$ is precisely $\Psi_{G\backslash e_1}$ by the contraction-deletion identity $\Psi_G = \alpha_1 \Psi_{G\backslash e} + \Psi_{G\q e}$.

$(ii)$.  By the same argument as in $(i)$, we have $\det(\Lambda_G^{i,i}) = \Psi_{G \backslash e_i} $ for $i=1,2$ and $\det (\Lambda_G^{12,12}) = \Psi_{G\backslash \{e_1,e_2\}}$.  Now  the `Dodgson identity' 
$$ \left(\det(\Lambda_{G}^{1,2})\right)^2  =  \Psi_{G \backslash e_1} \Psi_{G \backslash e_2} - \Psi_{G \backslash \{e_1,e_2\}} \Psi_G   $$
implies that $\det(\Lambda_{G}^{1,2}) $ is equal, up to a possible sign, to the `Dodgson polynomial'   $\Psi^{1,2}_G$  which satisfies the identical identity (but is defined using the graph matrix, and not the Laplacian). An expression for the latter  in terms of  the requisite forest polynomials follows from   \cite[\S2]{BrownYeats}. 
 \end{proof}

\section{Determinants of generalised graph Laplacians} \label{sect: detLambdaG}
Let $G$ be a connected graph as in \S\ref{sec: GenLaplace} with internal particle masses and  external momenta lying in 2-dimensional Euclidean space $\R^2 \cong \C$. 

\begin{thm}  \label{thm: detLambdaisXi} Let  $\widetilde{\Lambda}_G(\mu,m)$ be a generalised graph Laplacian for any choice of momentum routing $\mu_e \in \C^{E_G}$. 
It satisfies
\begin{equation} \label{detLambdaisXi}
\det \widetilde{\Lambda}_G(\mu, m) = \Xi_G(q,m)\ .
\end{equation}
\end{thm}

Now let $G$ be  as above with external momenta in $\R^4\cong \Qt$. 
\begin{thm}  \label{thm: detLambda4isXi} Let $\widetilde{\Lambda}_G(\mu, m)$ be a generalised quaternionic graph Laplacian with respect to a choice of momentum routing $\mu \in \Qt^{E_G}$.   It satisfies 
\begin{equation} \label{detLambda4isXi}
\mathrm{Det} \,  \widetilde{\Lambda}_G(\mu, m) = \left(\Xi_G(q,m)\right)^2\ .
\end{equation}
\end{thm}

\begin{rem} The right-hand side of \eqref{detLambdaisXi}  or \eqref{detLambda4isXi} only depends on external kinematics. It follows that the determinant of the generalised graph Laplacian with generic momentum variables $\mu_e$ only depends on the quantities 
\[ 
\mu_v = \sum_{s(e) = v} \mu_e -   \sum_{t(e) = v} \mu_e 
\]
where the sum is over all  internal edges incident to a vertex $v \in V(G)$. 
\end{rem}

We first prove theorem \ref{thm: detLambdaisXi} and use it to deduce theorem \ref{thm: detLambda4isXi}. 
It is instructive to prove   theorem  \ref{thm: detLambdaisXi} by direct computation of the determinant of the left-hand side, making  use of proposition \ref{prop: PhiGasForest}.
An alternative approach could exploit the asymptotic factorisation identities proven in \S \ref{sect: Canatinfinity}, which were shown in \cite[\S4]{Cosmic} to determine $\Xi_G(q,m)$ essentially uniquely. 

\subsection{Preliminary reductions} 
 It follows from the presentation 
\eqref{LambdaGblockmatrix} of $\widetilde{\Lambda}_G(\mu,m)$ in block matrix form that 
$$ \det \widetilde{\Lambda}_G (\mu,m)=  \left(\sum_{i=1}^N \alpha_i m_i^2 \right) \det \Lambda_G + \det \widetilde{\Lambda}_G (\mu, 0)$$
Since $\det(\Lambda_G) = \Psi_G$,  it is enough by  definition of $\Xi_G(q,m)$ to prove 
\eqref{detLambdaisXi} 
 in the case when all masses are zero. Thus,  theorem \ref{thm: detLambdaisXi} is equivalent to:
\begin{equation} \label{in proof: detLambdaisphi} 
\det \widetilde{\Lambda}_G  (\mu,0)  = \Phi_G(q)\ .
\end{equation} 
Now recall that $\widetilde{\Lambda}_G  (\mu,0)$ takes the form:
\[
  \widetilde{\Lambda}_G(\mu,0)  = \left(
  \begin{array} {ccc|c}
  &  & &   c_1(\alpha,  \mu) \\
  &  \Lambda_G & &  c_2(\alpha, \mu)  \\
   &  & &  \vdots \\
   \hline
 c_1(\alpha,  \overline{\mu})    & c_2(\alpha, \overline{\mu}) & \ldots  & \sum_{i=1}^N \alpha_i \mu_i \overline{\mu_i}
  \end{array}  \right)  \ . 
  \]
  By \S\ref{sect: ChangeofBasis}, its determinant does not depend on the choices of cycles  $c_i$. To verify \eqref{in proof: detLambdaisphi}, we use  proposition \ref{prop: PhiGasForest}  and compare the coefficients of $\mu_{e_1}  \overline{\mu}_{e_2}$ on both sides of \eqref{in proof: detLambdaisphi}.   First we treat the case $e_1=e_2 =e$. 

\begin{lem} \label{lem: casemuemubare}  Let $e$ be an edge of $G$, then the coefficients of $\mu_e \overline{\mu}_e$ on the left and right hand sides of \eqref{in proof: detLambdaisphi} agree. 
\end{lem} 
\begin{proof}
Consider the coefficient of $\mu_e \overline{\mu}_e$ in $\det \widetilde{\Lambda}_G  (\mu,0)$.   If $e$ is a bridge, then we can find  cycles $c_1,\ldots, c_h\in \Z^{E_G}$ representing a basis of  $H_1(G;\Z)$ which do not involve $e$ at all. In this case, $\mu_e$ only appears in $\widetilde{\Lambda}_G  (\mu,0)$ in the bottom right hand corner, and the coefficient of $\mu_e \overline{\mu}_e$ is 
\[ 
\det (\Lambda_G) \alpha_e =\Psi_G \alpha_e  = \Psi_{G\q e} \alpha_e  \ . 
\]
If $e$ is not a bridge, then choose cycles $c_1,\ldots, c_h$ as in  lemma \ref{lem: choosecycles}. There are two possible types of  terms in an expansion of the determinant which can involve $\mu_e \overline{\mu}_e$, namely those involving the top right and bottom left corners $c_1(\alpha, \mu)$ and $c_1(\alpha, \overline{\mu})$, and those involving the bottom right hand corner as before. In total they contribute
$$\alpha_e \det(\Lambda_G)-  \alpha_e^2 \det(\Lambda_G^{1,1})   = \alpha_e \left(  \Psi_G - \alpha_e \Psi_{G\backslash e}\right)   = \alpha_e  \Psi_{G\q e}  \ .$$
The minus sign on the left comes from  $(-1)^h (-1)^{h-1}=-1$; the first equality is a consequence of lemma \ref{lem: MinorsofLambdaG} (i); and the second equality is the contraction-deletion identity: $\Psi_G = \alpha_e \Psi_{G\backslash e}+ \Psi_{G \q e}$. 
In all cases,  we conclude  that the coefficient of $\mu_e \overline{\mu}_e$ in  $\det \widetilde{\Lambda}_G  (\mu,0)$ is $ \alpha_e \Psi_{G\q e}$, in agreement with  \eqref{PhiGasForest}.
\end{proof}

Before proceeding to the case $e_1, e_2$ distinct, we consider some boundary cases. 

\begin{lem}  \label{lem: deletetadpoles} Let $e$ be a tadpole in $G$. Then 
\[ 
\det \widetilde{\Lambda}_G(\mu,0) = \alpha_e \det \widetilde{\Lambda}_{G\backslash e}(\mu,0)  
\]
\end{lem}

\begin{proof} By lemma \ref{lem: choosecycles},  choose the first cycle in a basis $c_1,\ldots, c_h$ of $H_1(G)$ to be  $c_1= e$, and all other $c_i$ to be independent of $e$.  The matrix $\widetilde{\Lambda}_G(\mu,0)$ has the form 
\[
  \widetilde{\Lambda}_G(\mu,0)  = \left(
  \begin{array} {c|ccc|c}
   \alpha_{e}  &  0  &\ldots &  0 &   \alpha_{e} \mu_{e} \\ \hline
0  &  & & &  c_2(\alpha, \mu)  \\
  \vdots &  & \Lambda_{G\backslash e}    &&   \vdots \\
 0 &  & &&   c_h(\alpha, \mu) \\
   \hline
\alpha_{e}  \overline{\mu}_{e}    & c_2(\alpha, \overline{\mu}) & \ldots  & c_h(\alpha, \overline{\mu})&  \sum_{i=1}^N \alpha_i \mu_i \overline{\mu_i}
  \end{array}  \right)  \ . 
  \]
The result follows by subtracting $\overline{\mu}_{e}$ times the first row from the final row,  and performing an expansion of the determinant along the first column. 
\end{proof} 

Next we verify that the theorem is true in the massless case for  `banana'  graphs, i.e.,  graphs with  exactly  two vertices joined by $\geq 3$ internal edges. 

\begin{example} \label{example:detLambdaforBanana} 
Let $G$ have  two vertices $v_1,v_2$, and  $n \geq 3$ edges. Orient them so that they all  have source $v_1$ and target $v_2$, except for $e_n$, which has source $v_2$ and target $v_1$. Consider the basis of cycles:
\[
c_1 = e_1 + e_n \  , \ c_2 = e_2 + e_n \  , \  \ldots \ , \  c_h = e_{n-1} + e_n 
\]
where $h= n-1.$  With these conventions, the graph Laplacian is 
   \[
 \Lambda_G = 
 \left(
  \begin{array} {cccc}
   \alpha_{1} + \alpha_n  &   \alpha_n  &\ldots &   \alpha_n \\ 
  \alpha_n &   \alpha_2 +\alpha_n & \ldots &  \alpha_n  \\
  \vdots &  &   \ddots   & \vdots  \\
 \alpha_n &  \ldots  &   \alpha_n & \alpha_{n-1} + \alpha_n   \\
  \end{array}  \right)  \ .  
 \] 
 It follows from the explicit description of $\widetilde{\Lambda}_G(\mu,m)$ in the form \eqref{LambdaGblockmatrix} that 
 the coefficient  of $\mu_1 \overline{\mu}_2$ in $\det \widetilde{\Lambda}_G(\mu,0)$ is  
  $\alpha_1 \alpha_2 \det \Lambda_G^{1,2}$, 
  where 
\[ 
\Lambda_G^{1,2}  =
\left(
  \begin{array} {cccc}
    \alpha_n  &   \alpha_n  &\ldots &   \alpha_n \\ 
  \alpha_n &   \alpha_3 +\alpha_n & \ldots &  \alpha_n  \\
  \vdots &  &   \ddots   & \vdots  \\
 \alpha_n &  \ldots  &   \alpha_n & \alpha_{n-1} + \alpha_n   \\
  \end{array}  \right)  \ .  
\]
By subtracting the first row from all other rows, we see that 
$\det \Lambda_G^{1,2} = \alpha_3\ldots\alpha_n$, and therefore the coefficient of $\mu_1\overline{\mu}_2$ in $\det(\Lambda_G(\mu,0))$ is $\alpha_1\ldots \alpha_n$. By symmetry, the same holds for any $\mu_i \overline{\mu}_j$, where $i\neq j$.  This agrees with the coefficient of $\mu_i \overline{\mu}_j$ in $\Phi_G(q)$, by its definition (there is only one spanning two-tree). 
The case when $i=j$ follows from lemma \ref{lem: casemuemubare}, and  it follows that 
$\det  \widetilde{\Lambda}_G(\mu,0) = \Phi_G(q) $
as  claimed.
\end{example}

\subsection{Proof of theorem \ref{thm: detLambdaisXi} } 

Now let $e_1,e_2\in E(G)$ be  two distinct edges and consider the coefficient of $\mu_{e_1} \overline{\mu}_{e_2}$ in $ \det( \widetilde{\Lambda}_G(\mu,0)) $. 
Suppose first of all that one of the edges, say $e_1$, is a bridge. By the argument given above,  no terms in the  matrix   $ \widetilde{\Lambda}_G(\mu,0) $  depend on $\mu_{e_1}$ except for the bottom right-hand corner, and so  the coefficient of $\mu_{e_1} \overline{\mu}_{e_2}$ in $ \det( \widetilde{\Lambda}_G(\mu,0)) $ is zero. Now since $e_1$ is a bridge, it follows that $e_2$ lies in just one of the two connected components of $G \backslash e_1$, and therefore:
\[ 
  \phi_G^{\{s(e_1), s(e_2)\}, \{t(e_1), t(e_2)\}}    =    \phi_G^{\{s(e_1), t(e_2)\}, \{t(e_1), s(e_2)\}} = 0 \ . 
  \] 
 Now suppose that neither $e_1$ nor $e_2$ is a bridge, yet $G\backslash \{e_1,e_2\}$ has two connected components. In particular,  $e_2$ is a bridge in $G\backslash \{e_1\}$.   In this case we may choose the cycles $c_1, \ldots, c_h$ such that $c_2, \ldots, c_h$ are cycles in $G\backslash \{e_1\}$ which   do not depend on $e_2$. Thus we can assume that only cycle $c_1$ involves either  $e_1$ or $e_2$, and furthermore  we can direct the edges $e_1,e_2$ in such a way that their coefficient in $c_1$ are both $1$.   In this case, the coefficient of $\mu_{e_1} \overline{\mu}_{e_2}$  in $ \det( \widetilde{\Lambda}_G(\mu,0)) $  is 
 \[ 
  (-1)^{h}(-1)^{h-1} \alpha_{e_1} \alpha_{e_2} \det  \widetilde{\Lambda}^{1,1}_G(\mu,0) = - \alpha_{e_1}\alpha_{e_2} \Psi_{G\backslash \{e_1\}}
 \]
 Since the sources of $e_1,e_2$ lie in different components of  $G\backslash \{e_1,e_2\}$, we have 
 \[
   \phi_G^{\{s(e_1), s(e_2)\}, \{t(e_1), t(e_2)\}} = 0\ .
 \] 
 There is a   one-to-one correspondence between  the set of  spanning forests with two components $T_1 \cup T_2$ such that $ \{s(e_1), t(e_2)\} \subset V(T_1)$ and  $\{t(e_1), s(e_2)\} \subset V(T_2)$, and the set spanning trees $T= T_1 \cup T_2 \cup e_2$ of $G\backslash \{e_1\}$. This implies that
 \[ 
 \phi_G^{\{s(e_1), t(e_2)\}, \{s(e_1), t(e_2)\}}  = \alpha_{e_1} \alpha_{e_2}    \Psi_{G\backslash \{e_1\}}
  \]
  and again the coefficient of $\mu_{e_1} \overline{\mu}_{e_2}$ in $\det  ( \widetilde{\Lambda}_G(\mu,0)) $ and $\Phi_G(q)$ are in agreement.

 Finally, consider the case when $G \backslash \{e_1, e_2\}$ is connected.  Choose cycles $c_1,\ldots, c_h$ as in lemma \ref{lem: choosecycles}. By  expanding  $\det \,\widetilde{\Lambda}_G(\mu, 0)$,      its coefficient of $\mu_{e_1} \overline{\mu}_{e_2}$    is 
$$  \alpha_{e_1} \alpha_{e_2} \det(\Lambda_G^{1,2})  =  \pm  \alpha_{e_1} \alpha_{e_2} \left(        \phi_G^{\{s(e_1), s(e_2)\}, \{t(e_1), t(e_2)\}}  -      \phi_G^{\{s(e_1), s(e_2)\}, \{t(e_1), t(e_2)\}}          \right)   \ ,   $$
by  lemma \ref{lem: MinorsofLambdaG}. By proposition \ref{prop: PhiGasForest},  the coefficient of $\mu_{e_1} \overline{\mu}_{e_2}$
 in $ \det  \widetilde{\Lambda}_G(\mu,0)$  is equal to that of 
 $\mu_{e_1} . \mu_{e_2}$ in $\Phi_{G}(q)$ up to a possible sign. It remains to determine the sign.

If the coefficient of $\mu_{e_1} \overline{\mu}_{e_2} $ vanishes in $ \det  \widetilde{\Lambda}_G(q,0)$ then there is nothing to prove.  Otherwise,  it is enough to select a single non-zero monomial  $m$ in the  coefficient of $\mu_{e_1} \overline{\mu}_{e_2} $   in $ \det  \widetilde{\Lambda}_G(\mu,0)$, and compare with the corresponding  coefficient in $\Phi_G(q)$. Such a monomial corresponds to  a unique spanning forest $\mathcal{F}=T_1 \cup T_2$ via the formula $m = \prod_{e\notin \mathcal{F}} \alpha_e$.  The sign is determined by setting  all $\alpha_e$ to $0$ for $e \in \mathcal{F}$, which corresponds to contracting all edges in $\mathcal{F}$.   Since 
\[
\widetilde{\Lambda}_{G/ \mathcal{F}} (\mu,0)=    \widetilde{\Lambda}_{G} (\mu,0)\Big|_{\alpha_e=0,\,  e\in \mathcal{F}}\ ,
\] 
it therefore suffices to prove the identity  for the graph  $G'=G / \mathcal{F}$. 
Because $\mathcal{F}$ has two components,   $G'$ has exactly  two vertices. We can assume that  $G' \backslash \{e_1,e_2\}$ is  connected since we have already considered the disconnected case.  
By lemma \ref{lem: deletetadpoles}, we can  remove all tadpoles  from $G'$, and thus assume that $G'$ is  a banana graph with two vertices. The result then follows from  example \ref{example:detLambdaforBanana}.

\subsection{Quaternionic case (proof of theorem \ref{thm: detLambda4isXi})}
It follows from  proposition \ref{prop: Hermmatrices}  that there exists a well-defined polynomial $F$ in  the $\alpha_i, m_i$ and the components $\lambda_e, \nu_e$ of the momenta $\mu_e = \lambda_e + \Qj \nu_e$,  and their conjugates, such that 
\[ 
\mathrm{Det} (\widetilde{\Lambda}_G(\mu,m) ) = \det ( \chi_{\widetilde{\Lambda}_G(\mu,m)}  ) = F^2 \ .
\]
By inspection of \eqref{ExpandedLaplacianBlockMatrix}, the middle term  is at most quadratic in the mass squares $m_e^2$, and therefore $F$ may be written in the form 
\[
 F = F_0 + \sum_{e\in E_G} m_e^2 F_e\ . 
\]
More precisely, by expanding $\det (\chi_{\widetilde{\Lambda}_G(\mu,m)} ) $  using the representation \eqref{ExpandedLaplacianBlockMatrix},  one sees that the quadratic terms in the mass squares all  arise from the term 
\[
 X^2 \det \chi_{\Lambda_G}  = (X \det(\Lambda_G))^2
 \]
using  proposition \ref{prop: Hermmatrices} (ii).  By comparing terms in $m_e^4$,  it follows from this that 
\[  
\sum_{e\in E_G} m_e^2 F_e =    \left( \sum_{e\in E_G} m_2^2 \alpha_e \right) \det \left(\Lambda_G\right) =     \left( \sum_{e \in E_G} m_2^2 \alpha_e \right)\Phi_G
\]
and therefore 
\[ 
\mathrm{Det} (\widetilde{\Lambda}_G(\mu,0) ) = F_0^2 \ .
\]
We are thus reduced to proving the case when all masses are zero.

By inspection of the matrix \eqref{ExpandedLaplacianBlockMatrix},  the determinant is of degree at most four in the momentum components $\lambda_e, \nu_e, \overline{\lambda}_e, \overline{\nu}_e$, and therefore $F$ is of degree at most two in them. It suffices to check that  for any two edges $e,f$, the  coefficient of $\lambda_e \overline{\lambda}_f$ (or $\overline{\nu}_e \nu_f$, etc)  in $F$ coincides with that of $\Phi_G(q)$. Since $F$ does not depend on the choice of momentum routing, we can assume that the momenta through edges $e, f$ are purely complex, since they lie in a copy of $\R^2\cong \C $ inside $\R^4 \cong  \Qt$. In other words, we may assume that  $\nu_e = \nu_f=0$ is zero. Since the coefficient of $\lambda_e \overline{\lambda}_f$ in $F$ does not depend on the other momentum variables, we can assume that the latter are complex as well.  Now if $\nu_e=0$ for all edges $e$, the complex adjoint Laplacian 
$
  \chi_{\widetilde{\Lambda}_G(\mu,m)}
$ 
is simply the image under $\chi$ of the generalised Laplacian with complex momenta. Thus theorem \ref{thm: detLambda4isXi} follows from theorem \ref{thm: detLambdaisXi} and proposition \ref{prop: Hermmatrices} (ii).

\section{Invariant forms associated  to graphs} \label{section: invariantformsofgraphs} 
We define  invariant differential forms associated to generalised graph Laplacians.

\subsection{Invariant forms} \label{sect: invariantforms}
Let us define a graded  exterior algebra 
\[ 
B_{\can} = \bigwedge \bigoplus_{k\geq 1}  \Q\,  \beta^{2k+1}
\]
in which each element $\beta^{2k+1}$ has degree $2k+1$. It is a graded-commutative Hopf algebra with respect to the coproduct
\begin{equation} 
\Delta_{\can}  :  B_{\can} \To   B_{\can}\otimes B_{\can} 
\end{equation}
for which the generators $\beta^{2k+1}$ are primitive:
\[
\Delta_{\can} \left(\beta^{2k+1}\right)  =  1 \otimes \beta^{2k+1} + \beta^{2k+1} \otimes 1  \ . 
\]
We shall call  a primitive  element \emph{odd} or \emph{even} depending on whether $k$ is odd or even. 
The Hopf subalgebra generated by the even forms $\beta^{4k+1}$ was denoted by $\Omega_{\can} \subset B_{\can}$ in  \cite{CanonicalForms}.
Let $R= \bigoplus_{n\geq 0} R^0$ be a graded-commutative, unitary,  differential graded algebra with differential $d: R^n\rightarrow R^{n+1}$ of degree $1$. 
\begin{defn} For any invertible matrix $X \in \mathrm{GL}_h(R_0)$, let us denote by 
\[
\beta^{2k+1}_X = \tr \left( \left( X^{-1} dX \right)^{2k+1}  \right)  \quad \in R^{2k+1} \ .
\]
By taking exterior products, we can define $\beta_X$ for any $\beta \in B_{\can}$. 
\end{defn} 
We briefly recall  properties of these forms. Self-contained proofs of these facts, which are classical, may be found in \cite{CanonicalForms}. 
\begin{enumerate}
\item   The forms $\beta^{2k+1}_X$ are closed:  $d \beta^{2k+1}_X =0 $
\vspace{0.05in}

\item  One has $ \beta^{2k+1}_{X^T}  = (-1)^k \beta^{2k+1}_X \ . $ In particular, the odd primitive forms vanish on the space of symmetric matrices. 

\vspace{0.05in}
\item The forms $\beta^{2k+1}_X$ are bi-invariant:
\[
\beta^{2k+1}_{AXB} = \beta^{2k+1}_X 
\]
 for any matrices $A, B \in \GL_h(R_0)$  which are constant, i.e.,  $dA= dB=0$.

\vspace{0.05in}
\item The primitive forms are additive:  $\beta^{2k+1}_{X_1 \oplus X_2} = \beta^{2k+1}_{X_1} + \beta^{2k+1}_{X_2}$  with respect to direct sums. It follows for general $\beta \in B_{\can}$ that
\[
\beta_{X_1 \oplus X_2 }  = \sum   \beta'_{X_1} \wedge \beta''_{X_2}  
\]
where we write $\Delta_{\can} \beta = \sum \beta' \otimes \beta''$. 

\vspace{0.05in}\item Let $\lambda \in \left(R^0\right)^{\times}$ be invertible. Then  
\[ 
\beta_{\lambda X} = \beta_X \quad \hbox{ for any } \ \beta \in B_{\can}
\]

\vspace{0.05in}\item   If $X$ has rank $h$, then 
\[ 
\beta^{2k+1}_{X} =  0 \qquad  \hbox{ if } \quad k\geq h\ .
\]

\end{enumerate}

\subsection{Canonical forms associated to graphs}
The canonical forms introduced in \cite{CanonicalForms}, which do not depend on kinematics,  will henceforth be called `of the first kind' since their denominators involve the first Symanzik polynomial.

\begin{defn}  Let $\omega \in B_{\can}$.  For any connected graph $G$ with internal particle masses and external momenta in $\R^2\cong \C$,  define:
\vspace{0.05in}

\emph{(i).} The \emph{first canonical form} to be 
\begin{equation}
\omega_G =     \omega_{\Lambda_G} \ , 
\end{equation} 
where $\Lambda_G$ is any choice of graph Laplacian matrix. 
\vspace{0.05in}

\emph{(ii).}   Define the  \emph{second  canonical form}  to be 
\begin{equation}
\varpi_G  (\mu,m)= \omega_{\widetilde{\Lambda}_G(\mu,m)}
\end{equation}
where $\widetilde{\Lambda}_G(\mu,m)$ is any choice of generalised graph Laplacian matrix.  It depends \emph{a priori} on the external momenta $\mu_e$ and particle masses $m_e$. 
\vspace{0.05in}

\emph{(iii).} Define a \emph{mixed canonical form}, or simply \emph{canonical form} to be any polynomial in the canonical forms $\omega_G$, $\varpi_G$ of the first or second kinds.
\end{defn} 
The case when momenta lie in $\R^4$ is discussed below. 

Since the graph Laplacian $\Lambda_G$ is symmetric,  odd canonical forms of the first kind are identically zero.
This is not the case for the forms of the second kind. A mixed canonical form is abstractly represented, therefore, by an element of 
\begin{equation} \label{Ocan} 
\Ocan_{\can} =   B_{\can} \otimes \Omega_{\can}
\end{equation} 
The following theorem summarizes some basic properties of these forms. 

\begin{thm} The primitive  canonical graph forms are well-defined, projectively invariant closed differential forms of the shape: 
\begin{eqnarray} 
\omega^{4k+1}_G  & =   & \beta^{4k+1}_{\Lambda_G} =   \frac{ N_1(\alpha, d\alpha) } {\Psi_G^{k+1}} \nonumber \\ 
\varpi^{2k+1}_G (q,m)  & =  &   \beta^{2k+1}_{\widetilde{\Lambda}_G(\mu,m)} = \frac{ N_2(\alpha, d\alpha , m_e, \cq_v, \overline{\cq}_v)} {\Xi_G(q,m)^{k+1}} \nonumber 
\end{eqnarray} 
where $N_1, N_2$ are polynomials with rational coefficients.  Let 
\[X'_G= V(\Psi_G) \cup V(\Xi_G(q,m)) \subset \Pro^{E_G}\]
denote  the generalised graph hypersurface. For any canonical form $\omega\in\Ocan_{\can}$ of degree $n$, $\omega_G$ is a closed, projectively invariant form of degree $n$:
\[
\omega_G \in \Omega^n( \Pro^{E_G} \backslash X'_G)
\]
which, under the assumption of  generic momenta (\S\ref{Genericity}), is smooth on the open coordinate simplex $\sigma_G \subset   \Pro^{E_G}(\R)$ defined by $\alpha_1, \ldots, \alpha_{e_G}>0$. 
\end{thm} 
\begin{proof}
The fact that the forms are  closed and well-defined follows from bi-invariance \S\ref{sect: invariantforms}, as does projective invariance. 
It is clear from the definition that their denominators only involve a power of the determinant: the order of this power   follows from \cite{CanonicalForms}[Theorem 5.2], and the expression for the determinant  as a graph polynomial  follows from theorem \ref{thm: detLambdaisXi}.
The numerator $N_2$ is a priori only a polynomial in the momentum routing $\mu_e$, but again by bi-invariance, it is independent of the choice of momentum routing, and hence only depends on external momenta $\cq_v, \overline{\cq}_v$. 
The fact that $\omega$ is smooth on the open simplex follows from the assumption that kinematics are generic, and hence $\Xi_{G}(q,m)$ and $\Psi_G$ are strictly positive on $\sigma_G$. 
\end{proof}

\begin{rem} 
For any Hermitian matrix $X$, one has
\begin{equation}  \label{omegaofconjugate}
\overline{\beta^{2k+1}_{X}} = (-1)^k \beta^{2k+1}_{X}  \ .
\end{equation} 
As a consequence, $\varpi_G^{2k+1}$ is real if $k$ is even, and imaginary if $k$ is odd. 
In the case when $k$ is even, $\varpi_G$ will be a rational function of  the inner products $q_i. q_j$ of external momenta. In the case when $k$ is odd, $\varpi^{2k+1}_G$  may also depend on the quantities
\begin{equation}
\lambda_{ij}   = \cq_i \overline{\cq}_j-  \cq_j \overline{\cq}_i  =  \pm 2 \sqrt{ (q_i. q_j)^2 -  q_i^2 q_j^2} \ 
\end{equation} 
\end{rem}

\begin{rem}
If  $G$ has only two external momenta $q_1=-q_2$,   $\widetilde{\Lambda}_G(\mu,m)$ is equivalent to  a symmetric matrix because 
 we can find a momentum routing so that  the $\mu_e$ lie on the  real line spanned by $q_1$. 
In this case,  $\varpi^{2k+1}_G$ vanishes when $k$ is odd.
\end{rem} 

\subsection{Further properties of canonical forms}
The following properties of canonical forms are almost immediate from the definitions.

\begin{lem}

(i).  Let $\pi$ be an automorphism of $G$.  It acts on $R^0$ by permutation of the edge variables $x_e$.  For any canonical form of mixed type $\omega \in \Ocan_{\can}$,  
\[
 \pi^* \omega_G  = \omega_G \]

(ii). Let $\gamma \subset G$ be any subset of edges such that $\gamma$ neither  contains a loop, nor is mass-momentum spanning. Then 
\[
  \omega_G\Big|_{L_{\gamma}}  = \omega_{G/\gamma}\  \]
  where $L_{\gamma} \subset \Pro^{E_G}$ denotes the linear subspace defined by $\alpha_e=0$ for  $e\in \gamma$. 
\end{lem} 
\begin{proof}
The proof for a general canonical form are essentially the same as for the first canonical form (\cite[\S6.2]{CanonicalForms}).
\end{proof}
\subsection{Additional form in degree 1} 
The bi-invariant form 
\[ 
\beta^1_X = \tr(X^{-1} dX) = d \log \det(X) 
\]
is  excluded because it is poorly behaved.  In particular, the  forms 
\[ \omega^1_G = d \log \det(\Lambda_G)  = \frac{d\,\Psi_G}{\Psi_G}  , 
\qquad \hbox{ and } \qquad 
\varpi^1_G = d \log \det(\widetilde{\Lambda}_G(\mu,m))  = \frac{d\, \Xi_G(q,m)}{\Xi_G(q,m)} \]
are not projectively invariant since $\partial \omega^1_G=h_G$ and $\partial \varpi^1_G =h_{G}+1$, where $\partial$ denotes the Euler vector field.  Nonetheless, the combination
\begin{equation}  \label{odefn}
\o^1_G  =  \partial (\omega^1_G  \wedge \varpi^1_G )  = h_G \varpi^1_G - (h_G+1)\, \omega^1_G  
\end{equation}
is  projectively invariant since $\partial^2=0$,  which provides an additional form in degree one.    However,  it  has poles  at infinity and so products of canonical forms  involving $\o^1_G$ often  give  rise to divergent integrals, and will be excluded on this account. See however, \S \ref{secto1}, for an example of a convergent integral  associated to $\o^1_G$. 

\begin{rem} The significance of the form $\o^1_G$ is that it appears in the connection which defines the twisted cohomology underlying dimensional regularisation. For example, one may consider the Mellin transform
\[
\int_{\sigma_G}   \left( \frac{\Psi_G^{h+1}}{\Xi_G(q,m)^{h}} \right)^s \omega_G 
\]
 where the form $\omega_G$ in the integrand is to be viewed as a cohomology class with respect to the twisted connection  $\nabla(\omega) =  -  s \o^1_G \wedge \omega+ d \omega$. 
\end{rem} 

\subsection{Quaternionic forms}
Let $G$ be a connected graph with external momenta in $\R^4$, which we identify with quaternions via \S\ref{sect: QuatMomenta}.

\begin{defn}  \label{def: quatcanonicalform} For any $\omega \in \Omega^{\can}$,  define the \emph{quaternionic canonical form of the second kind} associated to $G$ to be 
\[
\varpi_G^{\Qt} (\mu,m) =   \omega_{\chi_{\widetilde{\Lambda}_G(\mu,m)}}
\]
where $\chi_{\widetilde{\Lambda}_G(\mu,m)}$ is the image under \eqref{chidef} of any choice of quaternionic Laplacian. 
\end{defn}

\begin{rem} \label{rem: Quatoddvanish} Since the  quaternionic graph Laplacian  is Hermitian, its complex adjoint $X= \chi_{\widetilde{\Lambda}_G(\mu,m)}$ satisfies the following identity, where $J_n$ was defined in \S \ref{QuaternionicMatrices},
\[ X  = J_n^{-1} X^T  J_n 
\] 
and is therefore similar to its transpose. For any such matrix $\beta^{2k+1}_X=0$ if $k$ is odd. For this reason only the even canonical forms are relevant here. 
\end{rem} 

The quaternionic forms of the second kind have very similar properties to the complex case. 
We briefly summarise these properties in the following theorem.

\begin{thm}
 Every quaternionic canonical form of degree $n$ of the second kind is a  well-defined,  closed, projectively-invariant differential form 
 \[ 
 \varpi_G^{\Qt}(\mu,m)  \quad \in \quad \Omega^{n} (\Pro^{E_G} \backslash X'_G)
 \]
 which is smooth on the open simplex $\sigma_G$ when momenta are generic \eqref{Genericity}.  It is invariant with respect to automorphisms of the graph $G$ and compatible with contraction of subgraphs which have no loops and are not mass-momentum spanning. 
\end{thm}

\begin{rem} The primitive quaternionic forms may be defined directly using the quaternionic trace instead of passing via complex adjoints, using the identity:
\[ 
 \mathrm{Tr}( (M^{-1} dM)^n) = \tr( (\chi_{M}^{-1} d \chi_M)^n) = \beta^n_{\chi_M}
\]
This uses the fact that 
 $M\in M_n(\Qt)$ is invertible if and only if $\chi_M \in M_{2n}(\C)$ is invertible \S \ref{QuaternionicMatrices}.
\end{rem}

\section{Canonical forms at infinity}  \label{sect: Canatinfinity} 
This section consists of technical results concerning the asymptotic behaviour of canonical forms in the limit when the generalised graph Laplacians become singular.

\subsection{Asymptotic behaviour of generalised Laplacian matrices} 

Let $I \subset E_G$ denote any subset of the set of edges of $G$. Consider the map 
\begin{eqnarray} 
s_I^* : \Q[\alpha_e, e\in E_G] & \To & \Q[z, \alpha_e, e\in E_G] \\
s_I^* (\alpha_e)  & = & z \alpha_e \qquad \hbox{ if } e \in I \ , \nonumber  \\ 
s_I^* (\alpha_e)  & = &  \alpha_e \qquad \hbox{ otherwise }\ .   \nonumber
\end{eqnarray} 
Suppose that $G$ has internal masses and external momenta in $\R^2 \cong \C$,  for the time being.
Recall that $\mathcal{E}_G$ was the extension defined in \eqref{ExtensionofH1G}.
\begin{lem} \label{lem: blockmatrix} 
 i).  Let $\gamma \subset E_G$ be a subgraph.  Then, writing matrices in block matrix form   with respect to a choice of splitting  
\[
\mathcal{E}_G \cong H_1(\gamma;\C) \oplus \mathcal{E}_{G/\gamma} 
\] the rescaled generalised graph Laplacian satisfies:
\begin{equation} \label{LambdaUV} 
s_{\gamma}^* \widetilde{\Lambda}_G (\mu,m) 
=
\left(
\begin{array}{cc}
  z \Lambda_{\gamma}  &  zB    \\
zC   &  D    
\end{array}
\right)
\end{equation}
where  $B,C,D$ are matrices with entries in $\C[z, \alpha_e, e\in E_G]$ and 
\[ 
D\equiv  \widetilde{\Lambda}_{G/\gamma}(\mu, m)  \pmod{z} \ .
\]

ii). Now let  $\gamma \subset E_G$ be a mass-momentum spanning subgraph. If we write  matrices in block matrix form   with respect to a splitting  
\[
 \mathcal{E}_G \cong H_1(G/\gamma;\C) \oplus \mathcal{E}_{\gamma} \ , 
\] then the  rescaled generalised graph Laplacian satisfies:
 \begin{equation} \label{LambdaIR} 
s_{\gamma}^* \widetilde{\Lambda}_G (\mu,m)
=
\left(
\begin{array}{cc}
  A  &  zB    \\
zC   &  z \widetilde{\Lambda}_{\gamma}  (\mu, m)  
\end{array}
\right) \quad \hbox{ where } \quad A\equiv  \Lambda_{G/\gamma}  \pmod{z} 
\end{equation} 
for some matrices $A,B,C$ with entries in $\C[z, \alpha_e, e\in E_G]$.
\end{lem} 

\begin{proof}  
\emph{i).}    The short exact sequence 
\[  
0 \To H_1(\gamma; \Z) \To H_1(G; \Z) \To H_1(G/\gamma; \Z) \To 0 \  
\]
defines a filtration $0 \subset H_1(\gamma; \Z) \subset H_1(G; \Z)$. 
We may choose representatives $c_1,\ldots, c_g \in \Z^{E_G}$  for a basis of $H_1(G;\Z)$ with the property that  $c_1,\ldots, c_h \subset \Z^{E_{\gamma}}$ represent a basis for $H_1(\gamma; Z)$. A choice of momentum routing $\mu$ defines a splitting $\mathcal{E}_G= H_1(G;\C) \oplus \C$ of \eqref{ExtensionofH1G}, and hence we may write 
\[ 
\mathcal{E}_G  = H_1(G;\C) \oplus \C  \cong  H_1(\gamma;\C) \oplus  H_1(G/\gamma;\C) \oplus \C  \cong H_1(\gamma;\C) \oplus  \mathcal{E}_{G/\gamma}\ . 
\]
The statement \eqref{LambdaUV} follows from the general form \eqref{LambdaGblockmatrix} of the  matrix $\widetilde{\Lambda}_G(\mu,m)$ using the fact that   $s_{\gamma}^* c_i(\alpha, \mu) = z c_i(\alpha,\mu)$ for $1 \leq i \leq h$,  and  the fact that 
\[
s_{\gamma}^*  \widetilde{\Lambda}_G(\mu,m)\Big|_{z=0} =   \widetilde{\Lambda}_G(\mu,m)\Big|_{\alpha_e=0, e\in \gamma} = \widetilde{\Lambda}_{G/\gamma}(\mu,m)
\]

\emph{ii).}  Now let $\gamma$ be mass-momentum spanning. 
As  above, a choice of momentum routing $\mu$ defines a splitting 
$\mathcal{E}_G \cong H_1(G;\C) \oplus \C$, and we may write 
\[ 
\mathcal{E}_G \cong H_1(G/\gamma; \C) \oplus \mathcal{E}_{\gamma}  . \]
Let $c_1, \ldots, c_g$ be a system of representative cycles for a basis of $H_1(G;\Z)$, adapted to the filtration $H_1(\gamma; \Z) \subset H_1(G;\Z)$. 
Consider the form \eqref{LambdaGblockmatrix} of the matrix $\widetilde{\Lambda}_G(\mu,m)$ with respect to this basis. 
Since $\gamma$ is momentum spanning,  one has by definition $s_{\gamma}^* c_i(\alpha, \mu)  = z c_i(\alpha,\mu)$  for all $1 \leq i \leq g$. 
Since $\gamma$ is mass and momentum spanning, one has $s_{\gamma}^* X_G = z X_{\gamma}$, where $X_{\bullet}$ is  defined by formula \eqref{Xdef}.
The statement therefore follows from  $s_{\gamma}^* \widetilde{\Lambda}_{G}(\mu,m) \big|_{z=0} = \Lambda_{G/\gamma}$.
\end{proof} 
\begin{rem} By taking the determinant and using theorem \ref{thm: detLambdaisXi},  the  previous lemma gives an interpretation of the `asymptotic' factorisation formulae 
\begin{eqnarray}
 \Xi_G(q,m)  &\sim &  \Phi_{\gamma}   \Xi_{G/\gamma}(q,m)   \qquad  \hbox{ if } \gamma \subset E_G    \hbox{ not m.m.}    \nonumber \\
 \Xi_G(q,m) & \sim  &   \Xi_{\gamma}(q,m)  \Phi_{G/\gamma}     \qquad \hbox{ if }  \gamma \subset E_G \hbox{ m.m.}   \nonumber  
\end{eqnarray}
stated in  \cite[Theorem 2.7]{PeriodsFeynman} for generic kinematics.  This observation leads to  another  proof of theorem \ref{thm: detLambdaisXi} since the above factorisation formulae determine $\Xi_G(q,m)$ essentially uniquely \cite[Proposition 4.8]{PeriodsFeynman}.
\end{rem}

Note that in the case when $\gamma$ is  mass-momentum spanning, both equations   \eqref{LambdaUV} and \eqref{LambdaIR}  are simultaneously valid. Since in this case  $G/\gamma$ is scaleless, $\widetilde{\Lambda}_{G/\gamma}(\mu,m)$ has vanishing determinant and 
\eqref{LambdaUV} is not  relevant.

\subsection{Reminders on motives for graphs with generic kinematics} 
We assume generic kinematics \eqref{Genericity} throughout this section. 
Let us denote by 
\[ X_G  = V(\Psi_G) \cup V( \Xi_{G}(q,m) ) \subset \Pro^{E_G}\]
the union of graph hypersurfaces.   By blowing up linear subspaces as in \cite{BEK} corresponding to motic subgraphs (the notion of a motic subgraph is defined in \cite[Definition 3.1]{Cosmic}), one obtains  a space \cite[Definition 6.3]{PeriodsFeynman}
\begin{equation} \label{PiGblowup}  
\pi_G : P^G \To \Pro^{E_G}
\end{equation} 
which contains a distinguished simple normal crossing divisor $D$ whose image under $\pi_G$ is the union of  coordinate hyperplanes. 
The irreducible components of  $D$ are of two types: exceptional divisors  $D_{\gamma}$ corresponding to $\gamma$ a motic subgraph of $G$, and divisors $D_{e}$ where $e$ is an edge of $E_G$, which are strict transforms of the coordinate hyperplanes $\alpha_e=0$. 
We denote the strict transform of $X_G$ by $Y_G \subset P^G$. 

Every exceptional  irreducible  component $D_{\gamma}$  of $D$ admits a canonical isomorphism  $D_{\gamma} \cong P^{\gamma} \times P^{G/\gamma}$ which induces an isomorphism
\begin{equation}
 D_{\gamma} \backslash  ( D_{\gamma} \cap Y_G)  \ \cong \   \left(P^{\gamma}  \ \backslash \  Y_{\gamma} \right) \  \times  \ \left( P^{G/\gamma}  \  \backslash  \ Y_{G/\gamma}\right) \ .    
\end{equation}

\subsection{Canonical forms along exceptional divisors}
For any canonical form $\eta$,  denote  its pull-back to $P^G \backslash (Y_G \cup \mathcal{E})$ along     \eqref{PiGblowup} by 
\[ \widetilde{\eta}_G = \pi_G^* ( \eta_G)\ , \]
 where $\mathcal{E}$ denotes the exceptional divisor of $\pi_G$.  It could \emph{a priori} have poles along $\mathcal{E}$, but the following theorem shows that in fact it does not.

\begin{thm} \label{thm: Canformsalongdivisors} Let $\eta_G \in  \Ocan^n_{\can}$ be any  canonical form. 
Then   $\widetilde{\eta}_G$  has  no poles along $D$ 
and  extends to a  smooth form on  $P^G \backslash Y_G$, i.e.,
\[
  \widetilde{\eta}_G  \ \in \ \Omega^n ( P^G \backslash Y_G)  \ . 
\]
Its restriction to an irreducible boundary component $D_e$ of $D$ satisfies:
\[ 
\widetilde{\eta}_G\big|_{D_e} = \widetilde{\eta}_{G/e}
\]
if $D_e$ is the strict transform of the hyperplane $L_e  = V(\alpha_e)$ corresponding to a single  edge $e$ of $G$. 
The restriction of $\eta_G$ to exceptional divisors depend on its kind. For simplicity, assume that $\eta_G \in \{\omega_G, 
\varpi_G\}$ is primitive (the general case follows by multiplying primitive forms together), and that $\gamma\subset G$ is motic and $e_{\gamma}>1$.  
\vspace{0.05in}

(i). Suppose that $\eta_G=\omega_G$ is of the first kind. Then 
\begin{equation} \label{Restrict1stkindcore} 
\widetilde{\omega}_G\big|_{D_{\gamma}}   =   \widetilde{\omega}_{\gamma} \wedge 1 +  1 \wedge \widetilde{\omega}_{G/\gamma}  
\end{equation} 
if $\gamma$ is a core (bridgeless but not necessarily connected)  subgraph, and $\widetilde{\omega}_G\big|_{D_{\gamma}}$ vanishes if $D_{\gamma}$ is the exceptional divisor corresponding to a m.m. subgraph which is not core. 
\vspace{0.05in}

(ii). Suppose that $\eta_G = \varpi_G$ is of the second kind. Then 
\begin{eqnarray} \label{Restrict2ndkind}  
\widetilde{\varpi}_G\big|_{D_{\gamma}}  &= &    \widetilde{\omega}_{\gamma} \wedge 1 +  1 \wedge \widetilde{\varpi}_{G/\gamma}  \quad \hbox{ if } \gamma \hbox{ core, not m.m.}  \\ 
\widetilde{\varpi}_G\big|_{D_{\gamma}} &  = &   \widetilde{\varpi}_{\gamma} \wedge 1 + 1 \wedge \widetilde{\omega}_{G/\gamma}\quad  \hbox{ if } \gamma \hbox{ is m.m.} \nonumber \ . 
\end{eqnarray} 
\end{thm} 
\begin{proof}
The argument in \cite[Theorem 7.4]{CanonicalForms}  proves the case when $\eta$ is of the first kind and $\gamma$ is a core subgraph. 
The general result follows in an identical manner from a local description of the blow-ups in affine coordinates, combined with the  asymptotic formulae of  lemma \ref{lem: blockmatrix}. The essential part of the argument  is the calculation (\emph{loc. cit.}) of the behaviour of the bi-invariant form $\beta^n_X$ where $X$ is a matrix which is asymptotically block-diagonal as in \eqref{LambdaUV} or \eqref{LambdaIR}. 
\end{proof} 

\begin{rem}
An interesting feature of equations  \eqref{Restrict2ndkind} is that the right-hand side involves canonical forms of both the first and second kinds. In particular, if $\omega$ is primitive of degree $n\equiv 3 \pmod 4$ then the corresponding form of the first kind  $\omega_G$ vanishes, and the right-hand sides of \eqref{Restrict2ndkind} reduce to a single term.
\end{rem} 

The case of a non-primitive form is easily deduced from theorem \ref{thm: Canformsalongdivisors} by taking products.
For example,  for a general form $\omega$, equations \eqref{Restrict1stkindcore}, \eqref{Restrict2ndkind}  take the form 
\begin{eqnarray}  \label{RestrictionCanonicalFormsGeneral}
\widetilde{\omega}_G\big|_{D_{\gamma}}  & =   & \sum \widetilde{\omega}'_{\gamma} \wedge \widetilde{\omega}''_{G/\gamma}   \hbox{ if } \gamma \hbox{ is core},   \\ 
\widetilde{\varpi}_G\big|_{D_{\gamma}}  &= &   \sum \widetilde{\omega}'_{\gamma} \wedge  \widetilde{\varpi}''_{G/\gamma} \hbox{ if } \gamma \hbox{ is core and not m.m.,}  \nonumber \\
\widetilde{\varpi}_G\big|_{D_{\gamma}} & = &  \sum  \widetilde{\varpi}'_{\gamma} \wedge  \widetilde{\omega}''_{G/\gamma}  \hbox{ if } \gamma \hbox{ is m.m.},  \nonumber 
\end{eqnarray} 
where $\Delta_{\can} \omega= \sum \omega' \otimes \omega''$ is the coproduct of $\omega$.

\begin{rem} The differential form $\o^1_G$ in general has poles along boundary divisors.
Let $G$ be as above, and consider a core subgraph $\gamma \subset E_G$.
Then  by the partial factorisations for the first Symanzik polynomial  (or by  lemma \ref{lem: blockmatrix} i)), one has:
\[
s_{\gamma}^*  \omega^1_G =  s_{\gamma}^*  \,  d\log \Psi_G  =h_{\gamma} \frac{dz}{z} + d\log \Psi_{\gamma} + d \log \Psi_{G/\gamma} +O(z)
\]
\[
s_{\gamma}^*  \varpi^1_G =  s_{\gamma}^*  \,  d\log \Xi_{G}  =h_{\gamma} \frac{dz}{z} + d\log \Psi_{\gamma} + d \log \Xi_{G/\gamma} +O(z)\ .
\]
It follows that 
\[ 
s_{\gamma}^* \o^1_G  =  - h_{\gamma} \frac{dz}{z}  - \omega^1_{\gamma} + \o^1_{G/\gamma} +O(z)
\]
and so $\o^1_G$ has a simple pole along the exceptional boundary divisor  $D_{\gamma}$ (which is given by $z=0$ in local coordinates) and its residue depends only on the loop number of $\gamma$.  A similar argument shows that $\o^1_G$ also has poles along exceptional divisors $D_{\gamma}$, when $\gamma$ is mass-momentum spanning.  Differential forms of the type $\o^1_G \wedge \eta$, where $\eta$ is canonical, will therefore  in general  have poles along the divisor $D$, and will not be considered here. They naturally arise, however, when considering Stokes' formula for canonical forms in dimensional-regularisation.
\end{rem} 

\subsection{Quaternionic case}
The above arguments similarly apply to the case of a canonical form of the second kind  $\varpi^{\Qt}_{G}$ associated to a quaternionic generalised Laplacian, by replacing $\widetilde{\Lambda}_G(\mu,m)$  with the complex adjoint matrix $\chi_{\widetilde{\Lambda}_G(\mu,m)}$. Theorem \ref{thm: Canformsalongdivisors} holds verbatim (with the only difference that odd canonical forms of the second kind vanish in the quaternionic setting). In particular, $\varpi^{\Qt}_{G}$ has no poles along $D$ and extends to a smooth differential form on $P^G \backslash Y^G$. Its restriction to irreducible components of $D$ are given by \eqref{RestrictionCanonicalFormsGeneral}.

\section{Canonical integrals and relations from Stokes theorem} \label{sect: CanInt}
Throughout this section we assume generic (Euclidean) kinematics \eqref{Genericity}.
\subsection{Canonical integrals} Let $G$  be a connected graph  with an orientation $e_1\wedge \ldots \wedge e_{N} \in \left(\bigwedge \Z^{E(G)}\right)^{\times} $. It induces an orientation on the simplex $\sigma_G$.

\subsubsection{Complex momenta} Suppose that external momenta lie in $\R^2 \cong \C$.

\begin{thm}  \label{thm: converges} Let $\omega \in \Ocan_{\can}$ be a canonical form of degree $N+1$ and let $G$ be a connected graph with  $N$ edges. Then under assumption  \eqref{Genericity}, the integral
\begin{equation} \label{IGdef} 
I_{G}(\omega, q, m) = \int_{\sigma_G} \omega_G  
\end{equation} 
is finite and defines an  analytic function of the external kinematics.  
\end{thm} 
\begin{proof}
This is a corollary of theorem \ref{thm: Canformsalongdivisors}. Let $\widetilde{\sigma}_G = \overline{\pi^{-1}_G (\sigma_G)}$ denote the closure in the analytic topology of the pull-back of the open simplex $\sigma_G$ to $P^G$. One has
\[
I_{G}(\omega, q, m) = \int_{\sigma_G} \omega_G  = \int_{\widetilde{\sigma}_G} \widetilde{\omega}_G\ .
\] 
The right-most integral converges because of the fact (\cite{Cosmic}) that  $\widetilde{\sigma}_G$, which is compact,  does not meet $Y_G$, which is where the poles of   $\widetilde{\omega}_G$ are located. 
\end{proof} 

Reversing the orientation on $G$ reverses the sign of the canonical integral \eqref{IGdef}.
Furthermore, if $\tau: G \overset{\sim}{\rightarrow} G'$ is an isomorphism of oriented graphs (with external legs and masses, etc), then  since $\tau^* \omega_G = \omega_{G'}$, we have:
\[
I_{G}(\omega, q, m) =I_{G'}(\omega, q, m) \ . 
\]
Thus the  integrals \eqref{IGdef} are invariants of isomorphism classes of oriented graphs. 
\subsubsection{Quaternionic case} The previous theorem, and comments which follow, hold verbatim for graphs with momenta in 4-dimensional Euclidean space with
\[
I_G(\omega, q, m) = \int_{\sigma_G} \omega_G^{\Qt} \ .
\]
The only difference is that in this case we must assume $\omega \in \Omega_{\can}$ is even, since odd quaternionic canonical forms vanish by remark  \ref{rem: Quatoddvanish}.

\subsection{Stokes' formula}
In \cite{Cosmic}, it was shown  that one can extend the usual Connes-Kreimer coproduct to an enlarged coproduct
\[ \Delta (G) = \sum_{\gamma}  \gamma  \otimes G/\gamma
\]
where the sum is over all motic subgraphs of $G$.  It has two types of terms: if $\gamma\subset G$ is a core subgraph which is not \emph{m.m.}, then $\gamma$ is considered to be scaleless, and encodes ultraviolet divergences of the graph. If $\gamma \subset G$ is \emph{m.m.} then $G/\gamma$ is scaleless, and $\gamma$ encodes certain kinds of  infrared divergences of $G$.

\subsubsection{Complex momenta}

\begin{thm} \label{thm: Stokes} Let $\omega \in \Ocan_{\can}$ be a canonical form of degree $k$. Write its coproduct in the form $\Delta_{\can} \omega = \sum_i \omega'_i \otimes \omega''_i$.  Suppose that  $G$ is a connected oriented graph with external kinematics as above with $k+2$ edges. 
Then
\begin{equation}  \label{Stokes} 
0 = \sum_{e\in E_G} \int_{\sigma_{G /  e}} \omega_{G /  e} + \sum_i \sum_{\gamma}  \int_{\sigma_{\gamma}} (\omega'_i)_{\gamma}   \int_{\sigma_{G/\gamma}} (\omega''_i)_{G/\gamma} 
\end{equation} 
where the rightmost sum is over  motic  subgraphs $\gamma \subsetneq  G$  not including $G$ itself  (which are not necessarily connected) such that $e_{\gamma} = \deg \omega_i'+1>1$, and the orientations on $\sigma_{G / e}, \sigma_{\gamma} \times \sigma_{G/\gamma}$ are induced from those on $\sigma_G$. 
\end{thm}
\begin{proof} Identical to the proof of \cite[Theorem 8.5]{CanonicalForms}, using theorem \ref{thm: Canformsalongdivisors}.
\end{proof} 

Note that in the case when $\omega \in B^{\can}$ is purely of the second kind, one has 
\begin{multline}  \label{Stokes2} 
0 = \sum_{e \in E_G} \int_{\sigma_{G/ e}} \varpi_{G / e} + \sum_i  \sum_{\substack{\gamma  \,core \, \\   not\, m.m.}}  \int_{\sigma_{\gamma}} (\omega'_i)_{\gamma}   \int_{\sigma_{G/\gamma}} (\varpi''_i)_{G/\gamma}   \\
+  \sum_i  \sum_{\gamma\,   m.m.}  \int_{\sigma_{\gamma}} (\varpi'_i)_{\gamma}   \int_{\sigma_{G/\gamma}} (\omega''_i)_{G/\gamma} \ ,
\end{multline} {
where we have decomposed the set of motic subgraphs $\gamma$ of $G$, such that $1< e_{\gamma}<e_G$ into two types according to whether $\gamma$ is \emph{m.m.} or not.

Note that the Stokes' formulae above remain valid by analytic continuation to a larger region of kinematic space wherever the integrals are finite. 
\subsubsection{Quaternionic case} The statement is identical  on replacing $\omega$ with $\omega^{\Qt}$.

\begin{rem} In \cite[Theorem 8.5]{CanonicalForms}, the Stokes' formula is not quite correct as stated in the case when a graph $G$ has a tadpole $\gamma$. In that case,  the corresponding term was counted twice: once in the term $G/e$, where $e= \gamma$,  and again in the form $\gamma \times G/\gamma$.   The condition imposed above, that $e_{\gamma}>1$, fixes this problem.

Similarly,  if $\gamma$ is an \emph{m.m.} subgraph which has only a  single edge, then 
it only contributes a single term in the formula \eqref{Stokes}, which is counted in the left-most summation over all edges, and not in the right-hand double sum.
\end{rem}
 
\subsection{Canonical forms of compact type}
Let $h\geq 1$. Let us say that a canonical form $\omega\in B^{\can}$ is of compact type relative to $h$ if it is divisible by the primitive form $\omega^{2h-1}$ of degree $2h-1$, i.e., $\omega = \omega^{2h-1} \wedge \eta$ for some $\eta \in B^{\can}$.  It follows from the property (6)  stated in \S\ref{sect: invariantforms}, 
that $\omega$ vanishes on all matrices of rank $<h$.

\begin{cor} \label{cor:StokesCompactType}  Let $G$  be  as in  theorem \ref{thm: Stokes}   and let $\varpi \in B^{\can}$ be of the second kind of degree $e_G -2$, and  of compact type relative to $h_G+1$. Thus $\varpi = \varpi^{2 h_{G}+1} \wedge \eta$ for some $\eta \in B^{\can}$ of the second kind. If $G$ is core, then  formula \eqref{Stokes2} reduces to
\begin{equation} \label{GCrelation} 
0 = \sum_{e \in E_G} \int_{\sigma_{G /  e}} \varpi_{G /  e}  \ .\end{equation} 
\end{cor}
\begin{proof} 
Let $\gamma \subset G$ be a strict motic  subgraph with at least two edges.
Suppose that $\gamma$ is not  \emph{m.m.} and therefore   core.
 We need to check that  a term of the form 
\[ 
\int_{\gamma}  \omega'_{\gamma} \int_{G/\gamma} \varpi''_{G/\gamma}
\]
vanishes in the right-hand side of \eqref{Stokes2}. Either $\omega'$ or $\omega''$ is of compact type relative to $h_G+1$. If the former, then  $\omega'_{\gamma}= \omega'_{\Lambda_{\gamma}}$ vanishes since $h_{\gamma} < h_G+1$.  If $\omega''$ is of compact type, then
$\varpi''_{G/\gamma}= \omega''_{\widetilde{\Lambda}_{G/\gamma}}$ vanishes if $h_{G/\gamma} +1 < h_G+1$. This inequality holds since $h_{\gamma} + h_{G/\gamma}+1 = h_G+1$, and $h_{\gamma}>0$ because $\gamma$ is core.

 Now suppose that $\gamma$ is \emph{m.m.}.  
We need to check that  any term of the form 
\[ 
\int_{\gamma}  \varpi'_{\gamma} \int_{G/\gamma} \omega''_{G/\gamma}
\]
vanishes. If $\omega'$ is of compact type, then $\varpi'_{\gamma}= \omega'_{\widetilde{\Lambda}_{\gamma}}$ vanishes if $h_{\gamma} + 1 < h_{G}+1$. Since $\gamma$ is a strict subgraph, this follows from $h_{\gamma} < h_{G}$ since $G$ is core.  If $\omega''$ is of compact type, then 
$\omega''_{G/\gamma}= \omega''_{\Lambda_{G/\gamma}}$ vanishes since $h_{G/\gamma} < h_G+1$. 
\end{proof}

\section{Examples of canonical amplitudes (complex momenta)} \label{sect: ExamplesComplex}
For a graph $G$ with edges numbered $1,\ldots, N$, we write 
$$\Omega_G = \sum_{i=1}^N (-1)^i \, \alpha_i d \alpha_1 \wedge \ldots \wedge  \widehat{ d\alpha_i}\wedge  \ldots\wedge  d\alpha_N \ .$$
We consider canonical   differential forms in low degrees  and their integrals. For now, we consider only the case when external momenta lie in $\R^2 \cong \C$.

\subsection{Digression: the  divergent form $\o^1$}  \label{secto1}The  integrals involving  the differential form $\o^1$ diverge, and for that reason are not considered to be `canonical' forms according to our definitions. However, this form provides some interesting first examples, which may  on occasion lead  to convergent integrals.

Graphs which pair with $\o^1$ have two edges. 
Consider the moduli space  $\mathcal{M}^{\trop}_{1,2,2}$ of metric graphs with 1 loop, 2 external legs, and 2   distinct non-zero masses. 
 It has  two  metric graphs of maximal dimension (see figure \ref{fig: Massive2edge}). 
 The one on the right has a non-trivial subdivergence and so the integral of $\o^1$ is logarithmically divergent.  Therefore,  consider the graph  $G$ on the left. It is important that its edge masses $m_1, m_2$ be distinct, for otherwise, this graph has a symmetry which induces an odd permutation of its set of edges, implying that its canonical integral vanishes. 
  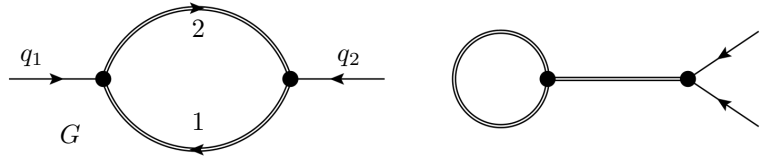
\begin{figure}[h]
\quad {\includegraphics[width=10cm]{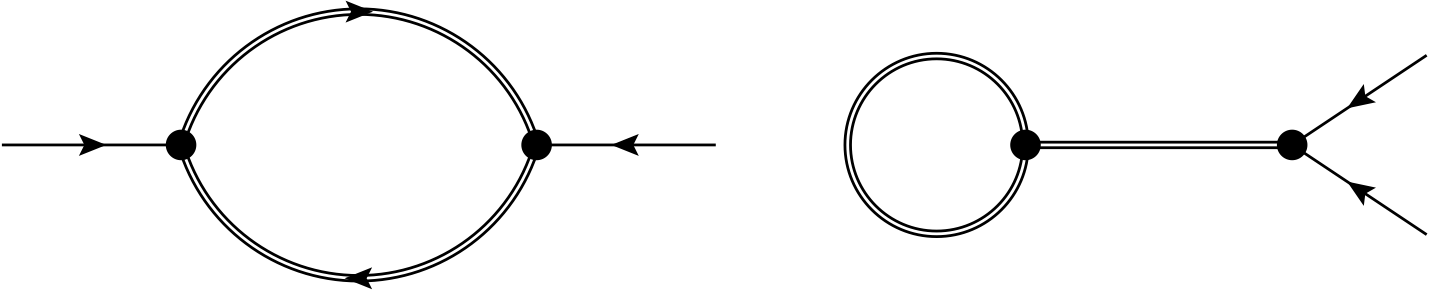}} 
\put(-280,36){$q_1$}\put(-160,36){$q_2$}
\put(-215,45){$2$}\put(-215,10){$1$}
\put(-265,5){$G$}
\caption{The two metric graphs of maximal dimension in  $\mathcal{M}^{\trop}_{1,2,2}$.
}
\label{fig: Massive2edge}
\end{figure}

Orient and number the edges as shown in the figure. 
A  representative for a generator of the  homology $H_1(G;\Z)$ is  $c_1 = e_1 +e_2$.  Momentum conservation implies that $q_2=- q_1 \in \R^2\cong \C$. Any momentum routing is given by a solution to:
\[  q_1 = \mu_2- \mu_1 \ . 
\]
The generalised graph Laplacian matrix with respect to these choices is:
\[
\widetilde{\Lambda}_G(\mu,m) = \begin{pmatrix} 
\alpha_1+\alpha_2 &  \mu_1 \alpha_1 + \mu_2 \alpha_2 \\
\overline{\mu}_1 \alpha_1+ \overline{\mu}_2 \alpha_2 &  \alpha_1 (\mu_1 \overline{\mu}_1+  m_1^2)  + \alpha_2 (\mu_2 \overline{\mu}_2+  m_2^2)  \\
\end{pmatrix}  
\]
The relevant graph polynomials are given by
\begin{eqnarray}
\Psi_G = \det(\Lambda_G) & =  & \alpha_1 + \alpha_2  \nonumber \\ 
\Phi_G (q)   & =  & q_1^2 \alpha_1 \alpha_2  \nonumber \\ 
\Xi_G(q,m) = \det(\widetilde{\Lambda}_G(
\mu,m)) & =  & q_1^2 \alpha_1\alpha_2+  (m_1^2 \alpha_1 + m_2^2 \alpha_2)  (\alpha_1 + \alpha_2 ) \ ,  \nonumber 
\end{eqnarray}
where we write $q_1^2$ for the Euclidean norm, which equals $\overline{\cq}_1 \cq_1$, where $\cq_1$ is $q_1$ viewed in  $\C$.   
Using \eqref{odefn}, the exceptional  form of degree $1$ is 
\[
 \o^1_G = \varpi_G^1 - 2 \,\omega^1_G =  d \log   \det(\widetilde{\Lambda}_G) - 2 \, d \log   \det(\Lambda_G) =
 d \log \left( \frac{\Xi_G(q,m)}{\Psi_G^2} \right)  \ . 
\]
Thus the associated  integral  is
\[
I_G(\o^1 , q, m) = \int_{\sigma_G}  \o^1_G = \int_{0}^{\infty}  d \log \left( \frac{\Xi_G(q,m)}{\Psi_G^2} \right)\Bigg|_{\alpha_2=1} d \alpha_1  
\]
where the last equality follows by computing the projective integral on the affine chart $\alpha_2=1$. The integral reduces  to 
\[
I_G(\o^1,q,m) \  = \   \left[ \log \left( \frac{\Xi_G(q,m)}{\Psi_G^2} \right) \right]^{\alpha_1=\infty, \alpha_2=1}_{\alpha_1=0,\alpha_2=1} \  =  \  \log \left( \frac{m_1^2}{m_2^2} \right)  
\]
which is an agreement with the renormalised amplitude for a massive bubble.  As expected, it vanishes when $m_1=m_2$ for the symmetry reasons mentioned above.

\subsection{Canonical amplitudes associated to $\varpi_3$}
We now turn to \emph{bona fide} canonical amplitudes, whose integrals are always finite by theorem \ref{thm: converges}. The canonical form of smallest degree is  $\varpi^3 \in B_{\can}$.   It pairs with connected graphs which have 4 edges, 
and  vanishes for graphs with fewer than 1 loop. Thus there are 3 cases:
\begin{itemize}
\item  $h_G=1, v_G=4$
\item  $h_G=2, v_G=3$
\item  $h_G=3, v_G=2$
\end{itemize} 
which we consider in turn for the maximal  number of  masses.

\subsubsection{$\mathcal{M}^{\trop}_{1,4,4}$} 
The massive box diagram (figure \ref{fig: MassBox}) is the unique graph which gives a cell of maximal dimension in $\mathcal{M}^{\trop}_{1,4,4}$.
   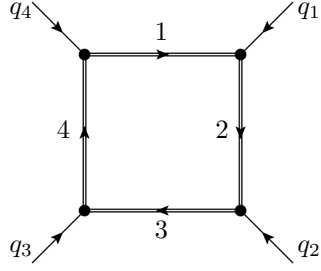
\begin{figure}[h]
\quad {\includegraphics[width=3.5cm]{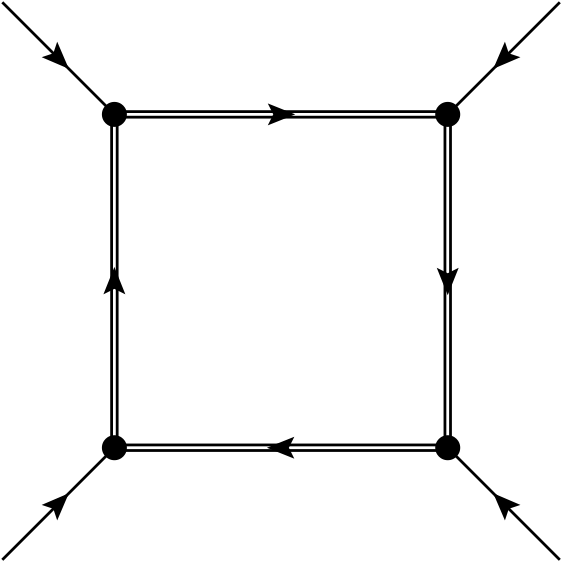}} 
\put(-53,85){$1$}\put(-30,48){$2$} \put(-53,10){$3$} \put(-90,48){$4$}\put(1,95){$q_1$}\put(1,5){$q_2$}\put(-108,5){$q_3$}\put(-108,95){$q_4$}
\caption{Massive box graph}
\label{fig: MassBox}
\end{figure}
Choose any complex numbers $\mu_e$ for each edge $e$ such that momentum
 conservation \eqref{localMC} holds:
$
\mu_{i+1} - \mu_{i} = \cq_i \ ,
$ 
where $q_i \in \R^2$, corresponding to $\cq_i \in \C$, is the external momentum at vertex $i$ subject to momentum conservation $\cq_1+\ldots +\cq_4=0$, and where we write $\mu_{5} =\mu_{1}$. The generalised graph Laplacian with respect to the cycle $c_1=e_1+\ldots +e_4$ is
\[ 
\widetilde{\Lambda}_G(\mu,m)  \ =  \  \begin{pmatrix} \alpha_1+ \alpha_2 + \alpha_3 + \alpha_4 & \mu_1 \alpha_1 + \mu_2\alpha_2 + \mu_3\alpha_3 + \mu_4 \alpha_4 \\ 
\overline{\mu}_1 \alpha_1 + \overline{\mu}_2\alpha_2 + \overline{\mu}_3\alpha_3 + \overline{\mu}_4 \alpha_4 & ( \mu_1\overline{\mu}_1 + m_1^2) \alpha_1 +\ldots  + ( \mu_4\overline{\mu}_4 + m_4^2)  \alpha_4 
\end{pmatrix}
\]
The associated graph polynomials are:
\begin{eqnarray} 
\Psi_G & = &   \alpha_1 + \alpha_2 + \alpha_3 + \alpha_4 \nonumber  \\ 
\Phi_G(q) & = & q_1^2 \alpha_1 \alpha_2 + q_2^2 \alpha_2 \alpha_3 + q_3^2 \alpha_3 \alpha_4  + q_4^2 \alpha_4 \alpha_1
+ (q_1+q_2)^2 \alpha_1 \alpha_3 + (q_1+q_4)^2 \alpha_2 \alpha_4 \nonumber \\ 
\Xi_G(q,m ) & = &\Phi_G(q)  + (   m_1^2 \alpha_1 + m_2^2 \alpha_2 + m_3^2 \alpha_3 + m_4^2 \alpha_4 ) \Psi_G  
\nonumber
\end{eqnarray} 
The canonical form  $\varpi^3_G$  is proportional to the Feynman differential form: 
\[
\varpi^3_{G}(q,m) =  3  N_G \,  \frac{ \Omega_G}{\Xi_G(q,m)^2}  
\]
where   the numerator has the symmetric form:
\[
 N_G= \det \begin{pmatrix} 
1 & 1& 1& 1  \\ 
\mu_1 & \mu_2 & \mu_3 & \mu_4 \\
\overline{\mu}_1 & \overline{\mu}_2 & \overline{\mu}_3 & \overline{\mu}_4 \\
m_1^2 + \mu_1 \overline{\mu}_1 & m_2^2 + \mu_2 \overline{\mu}_2& m_3^2+ \mu_3 \overline{\mu}_3 & m_4^2+ \mu_4 \overline{\mu}_4  \\ 
 \end{pmatrix} \ . 
 \] 
It can be re-expressed as a function of the external momenta:
\[
N_G= 2 i \,  \mathrm{Im} \left(  \cq_2 \overline{\cq}_3 \, m_1^2 -   \cq_3 \overline{\cq}_4\,m_2^2   +  \cq_4 \overline{\cq}_1\,m_3^2   -    \cq_1 \overline{\cq}_2\,m_4^2 + \cq_1 \cq_3 \overline{\cq}_2 \overline{\cq}_4  \right)  \ .
\]
In conclusion, we find that 
\[
I_G( 
\varpi^3, q, m) = \int_{L \mathcal{M}^{\trop}_{1,4,4} }  \varpi^3  = 3  N_G  \int_{\sigma_G}   \frac{ \Omega_G}{\Xi_G(q,m)^2} 
\] 
which is proportional to the Feynman integral for the massive box graph (viewed in $D=4$ spacetime dimensions).

\subsubsection{Examples from $\mathcal{M}^{\trop}_{2,3,4}$}  Figure \ref{Dunce} depicts  cells of dimension 4 on $\mathcal{M}^{\trop}_{2,3,4}$.
   \begin{figure}[h]
\quad {\includegraphics[width=10cm]{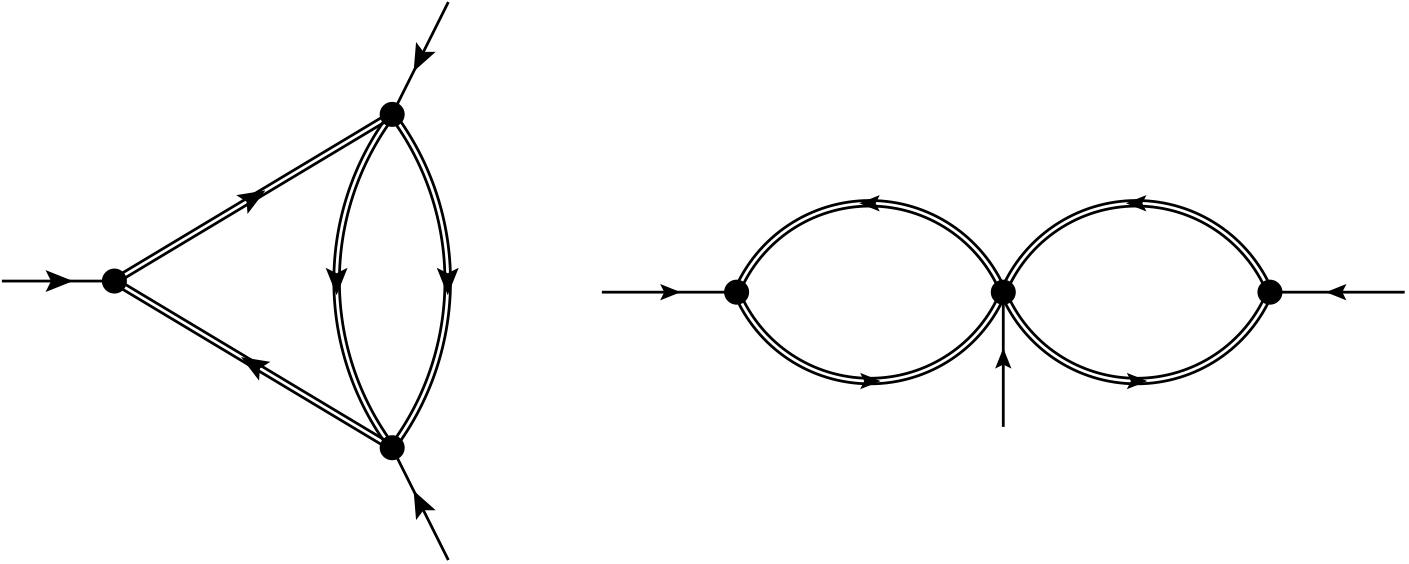}} 
\put(-295,60){$q_1$}\put(-190,100){$q_2$}\put(-190,10){$q_3$}
\put(-240,80){$1$}\put(-240,28){$2$}
\put(-225,55){$3$}\put(-202,55){$4$}
\caption{{\sc Left:} The dunce's cap has a non-trivial divergent subgraph given by the subgraph spanned by  edges $3,4$. 
{\sc Right:} A double bubble graph, upon which the canonical form vanishes. }
\label{Dunce}
\end{figure}
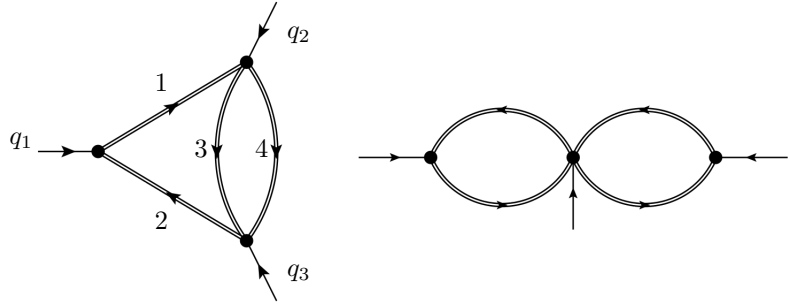
The graph on the right is an example of a 1-vertex join.  There is a similar graph in which the central vertex has no  incoming momentum, and the left-most  vertex has two incoming vertices. In both cases, the graph Laplacian has the form
\[
\begin{pmatrix}
\alpha_1 + \alpha_2 & 0  &  \mu_1 \alpha_1 + \mu_2 \alpha_2 \\ 
0 & \alpha_3 + \alpha_4   &  \mu_3 \alpha_3 + \mu_4 \alpha_4 \\ 
\overline{\mu}_1 \alpha_1 + \overline{\mu}_2 \alpha_2  & \overline{\mu}_3 \alpha_3 + \overline{\mu}_4 \alpha_4 & \sum_i (\mu_i \overline{\mu}_i + m_i^2 ) \alpha_i   
\end{pmatrix} 
\] 
and one verifies that $\varpi^3$ vanishes for this matrix. 
Therefore the only non-zero contribution is from the  dunce's cap, depicted on the left of figure \ref{Dunce}.  It  is one of the first examples of a graph with a non-trivial subdivergence (given by the subgraph spanned by edges $3,4$), so the associated Feynman integral diverges. By contrast, the canonical integral is necessarily finite and it is instructive to see why.

Since there are only 3 external legs, we can always assume that  the momenta are complex numbers $\cq_1, \cq_2, \cq_3$ subject to momentum conservation
$
 \cq_1 + \cq_2 + \cq_3 =0.
 $
 For the basis of cycles given by $c_1 = e_1+ e_3+ e_2$, $c_2= e_4-e_3$, and the orientations of edges as shown in figure \ref{Dunce}, the 
 generalised graph Laplacian takes the form:
 \[ 
 \widetilde{\Lambda}_G = \begin{pmatrix} 
 \alpha_1+ \alpha_2+\alpha_3 & -\alpha_3 & \alpha_1 \mu_1 + \alpha_2 \mu_2 + \alpha_3 \mu_3  \\
-\alpha_3 &  \alpha_3 +\alpha_4 &   \alpha_4 \mu_4 -  \alpha_3 \mu_3   \\
 \overline{\mu}_1 \alpha_1 +   \overline{\mu}_2 \alpha_2 +   \overline{\mu}_3 \alpha_3  
 & \overline{\mu}_4 \alpha_4 - \overline{\mu}_3 \alpha_3 &   \sum_{i=1}^4 \alpha_i(m_i^2 + \mu_i\overline{\mu_i})   
 \end{pmatrix} 
  \]
   where \eqref{localMC} takes the form:
   \begin{eqnarray} 
  \mu_1 - \mu_2 & =&  \cq_1 \nonumber \\ 
  \mu_3 + \mu_4- \mu_1 & =&  \cq_2 \nonumber \\ 
    \mu_3 + \mu_4- \mu_2 & =&  \cq_1+\cq_2 \ = \ - \cq_3 \ . \nonumber    
  \end{eqnarray} 
 One finds that 
 \begin{eqnarray} 
 \Psi_G & = &  (\alpha_1+ \alpha_2) (\alpha_3+ \alpha_4) + \alpha_3\alpha_4 \nonumber \\ 
 \Phi_G(q,m) & = &  \cq_1 \overline{\cq}_1 \left(\alpha_1 \alpha_2\alpha_3 + \alpha_1 \alpha_2 \alpha_4\right) + \cq_2 \overline{\cq}_2 \alpha_1\alpha_3\alpha_4 +  \cq_3 \overline{\cq}_3 \alpha_2\alpha_3\alpha_4  \ . \nonumber
 \end{eqnarray} 
A computation shows that   the canonical form of the second kind in degree 3 is 
\[ 
\varpi^3_G  =   3 (\cq_1 \overline{\cq}_2 - \cq_2 \overline{\cq}_1)  \,  \left( \alpha_3^2 m_3^2 - \alpha^2_4 m_4^2 \right)      \frac{\Omega_G}{\Xi^2_G(q,m)}\ ,
  \]
where we note that the factor  $ \cq_1 \overline{\cq}_2 - \cq_2 \overline{\cq}_1  = \cq_3 \overline{\cq}_1 - \cq_1\overline{\cq}_3$  is indeed symmetric,  by momentum conservation.
The  numerator vanishes at $\alpha_3 =\alpha_4=0$, to compensate the pole arising from the vanishing of $\Xi_G(q,m)$  along this locus.
 Note also that if $m_3=m_4$ then the integral of $\varpi^3_G$ will vanish, due to the fact that interchanging edges $3$ and $4$ induces an odd permutation on the edges of $G$.
In conclusion, for a suitable choice of orientation of $G$ one has 
\[
\int_{\sigma_G} \varpi^3_G =  3 (\cq_1 \overline{\cq}_2 - \cq_2 \overline{\cq}_1)  \, \int_{\sigma_G} \left( \alpha_3^2 m_3^2 - \alpha^2_4 m_4^2 \right)      \frac{\Omega_G}{\Xi^2_G(q,m)}\ .
\]
It is interesting to compare this to the renormalised amplitude associated to the graph $G$, which was discussed in \cite{BrownKreimer}.

\subsubsection{The space $\mathcal{M}^{\trop}_{3,2,4}$}
The associated graphs have at most two external momenta. By momentum conservation, they lie on a real line and so we may choose a real momentum routing for which  the generalised Laplacian $\widetilde{\Lambda}_G(q,m)$ will be symmetric. Therefore $\varpi^3_{G}$ vanishes, and there are no non-zero canonical integrals.

\subsection{Canonical amplitudes in degree 5} In degree 5 there exist  canonical forms $\omega^5$ and $\varpi^5$ of both  kinds.
There is a unique vacuum diagram, namely the wheel with 3 spokes  and no external legs, which pairs non-trivially with the former. Its 
amplitude is $6 \zeta(3)$, and is discussed in detail in \cite{CanonicalForms}. 

Henceforth we  shall focus on the form $\varpi^5$.  Since $\beta^5$ vanishes on matrices of rank $\leq 2$, the form $\varpi^5_G$ is zero  unless  $h_G\geq 2$.  There are several graphs  $G$ with 6 edges and two or more loops which could potentially pair non-trivially with $\varpi^G$,  but we shall only focus on a couple of examples for reasons of space. 

\subsubsection{Examples on $\mathcal{M}^{\trop}_{2,5,6}$} 
One shows by direct calculation that the two graphs depicted in  figure \ref{fig: 6edgenullgraphs} satisfy
$\varpi^5_G =0$. 
   \begin{figure}[h]
\quad {\includegraphics[width=10cm]{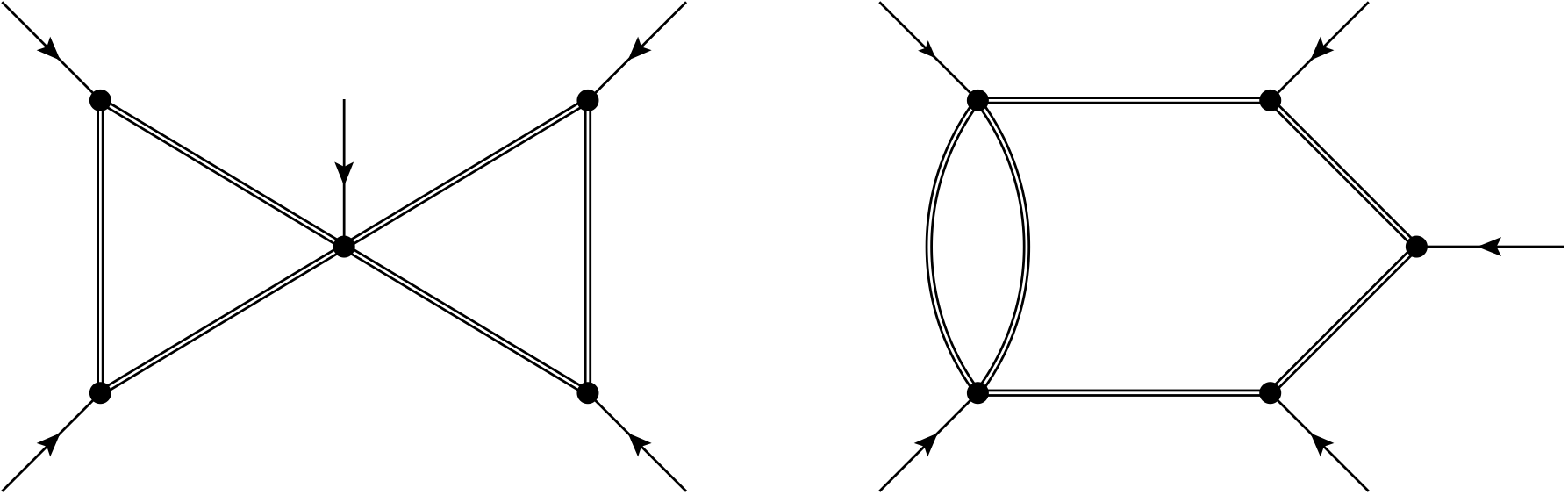}} 
\caption{Some graphs with 6 edges for which $\varpi^5_G$ vanishes}
\label{fig: 6edgenullgraphs}
\end{figure}

There is a single potentially non-zero canonical integral on $\mathcal{M}^{\trop}_{2,5,6}$, given by the 
 box-triangle graph   depicted in figure \ref{fig: BoxTriangle}. A computer calculation shows  that 
\[
\varpi^5_G =    \left(\frac{N_1}{\Xi_G(q,m)}   - 5 {N_2} \right)   \,  \frac{\Omega_G}{\Xi_G(q,m)^2}  
\]
where  $N_2$ is a certain polynomial in the $m_i^2$ and $q_i, \overline{q}_i$, 
and the polynomial $N_1$ factorises into a product of two terms
\[
N_1 = f(\alpha_1,\alpha_2,\alpha_3 ,q) g(\alpha_4,\alpha_5,\alpha_6,q, m) 
\]
where $f$ is a homogeneous linear form in $\alpha_1,\alpha_2,\alpha_3$ which does not depend on the masses and $g$ is of degree $2$ in the $\alpha_i$. 
   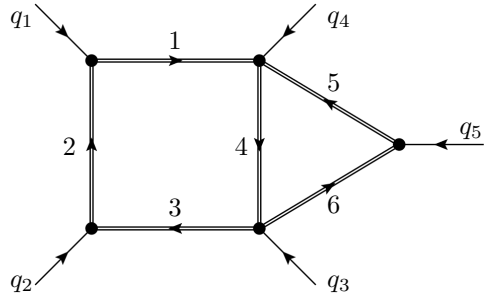
\begin{figure}[h]
\quad {\includegraphics[width=6cm]{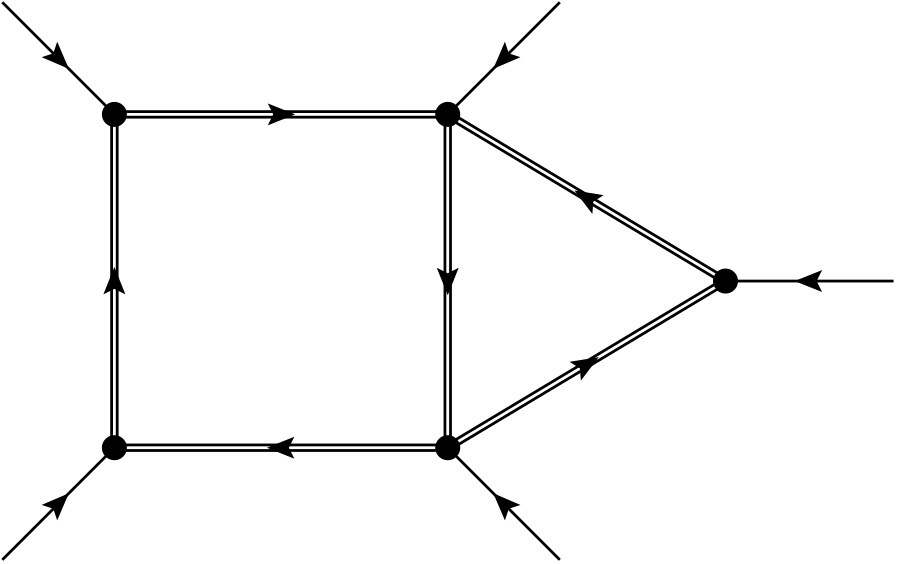}} 
\put(-180,100){$q_1$}\put(-60,100){$q_4$}\put(-10,58){$q_5$}
\put(-180,0){$q_2$}\put(-60,0){$q_3$}
\put(-120,90){$1$}\put(-120,27){$3$}\put(-160,50){$2$}\put(-95,50){$4$}\put(-60,74){$5$}
\put(-60,28){$6$}
\caption{The box-triangle graph}
\label{fig: BoxTriangle}
\end{figure}
The integral corresponding to the  term $N_2$ is the usual Feynman integral in 4 space time dimensions. It would be interesting to know if  the canonical integral admits a natural interpretation  in terms of Feynman integrals via integration by parts identities, or if it may be obtained from the usual Feynman integral by application of a natural differential operator in the external kinematics.

\section{Example of a Stokes relation for the massive  box diagram} \label{sect: Stokes rel} 

Consider the massive pentagon diagram  with arbitrary  masses  $m_1,\ldots, m_5 \in \R$ and external particle momenta in $q_1,\ldots, q_5\in  \R^2 \cong \C$ subject to momentum conservation.  Contracting each of its edges leads to five massive box diagrams.
 \begin{figure}[h]
\quad {\includegraphics[width=5cm]{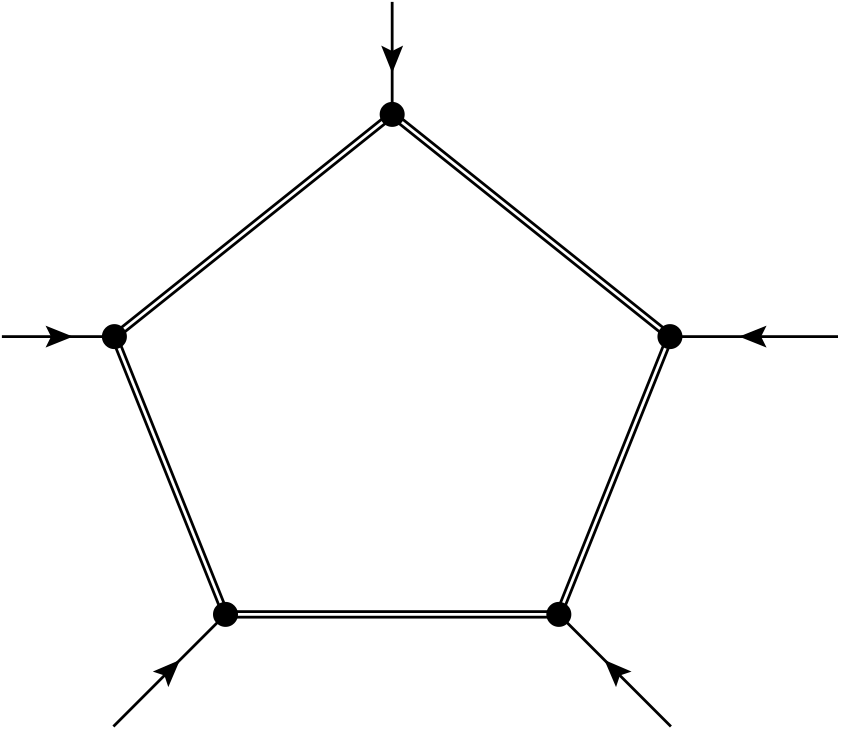}} 
\put(-80,130){$q_5$}\put(-50,90){$m_1$}\put(-115,90){$m_5$}
\put(-10,73){$q_1$}\put(-145,73){$q_4$}
\put(-110,45){$m_4$}\put(-55,45){$m_2$}
\put(-35,13){$q_3$}\put(-128,13){$q_2$}
\put(-80,10){$m_3$}
\caption{Massive pentagon diagram}
\label{fig: Pentagon}
\end{figure}
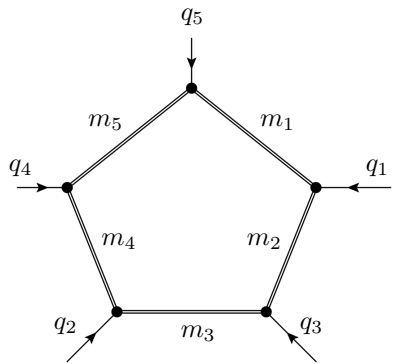

Since $\varpi^3$ is of compact type for 1-loop graphs, we may apply corollary \ref{cor:StokesCompactType} to deduce a five-term relation for the massive box diagram. Indeed \eqref{GCrelation} implies  that 
\begin{equation} \label{Boxfunctionalequation}
\sum_{i=1}^5 I_{G_i} (\varpi^3, q, m) =0 
\end{equation} 
where $G_1,\ldots, G_5$ are the five diagrams pictured in figure \ref{fig: 5term}. This relation can easily be checked in the massless case using Panzer's Hyperint \cite{HyperInt}.

 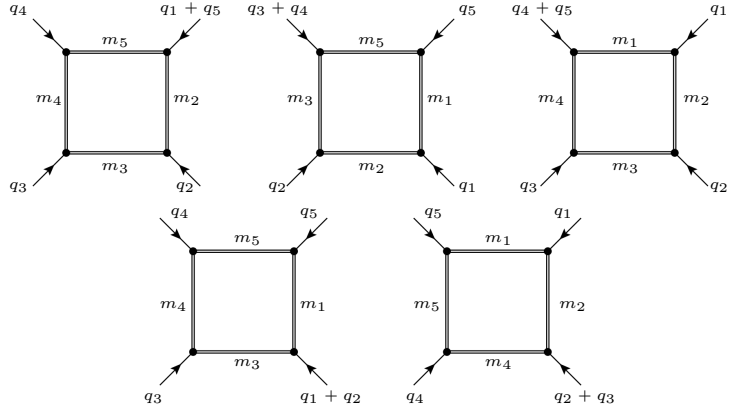
\begin{figure}[h]
\quad {\includegraphics[width=9cm]{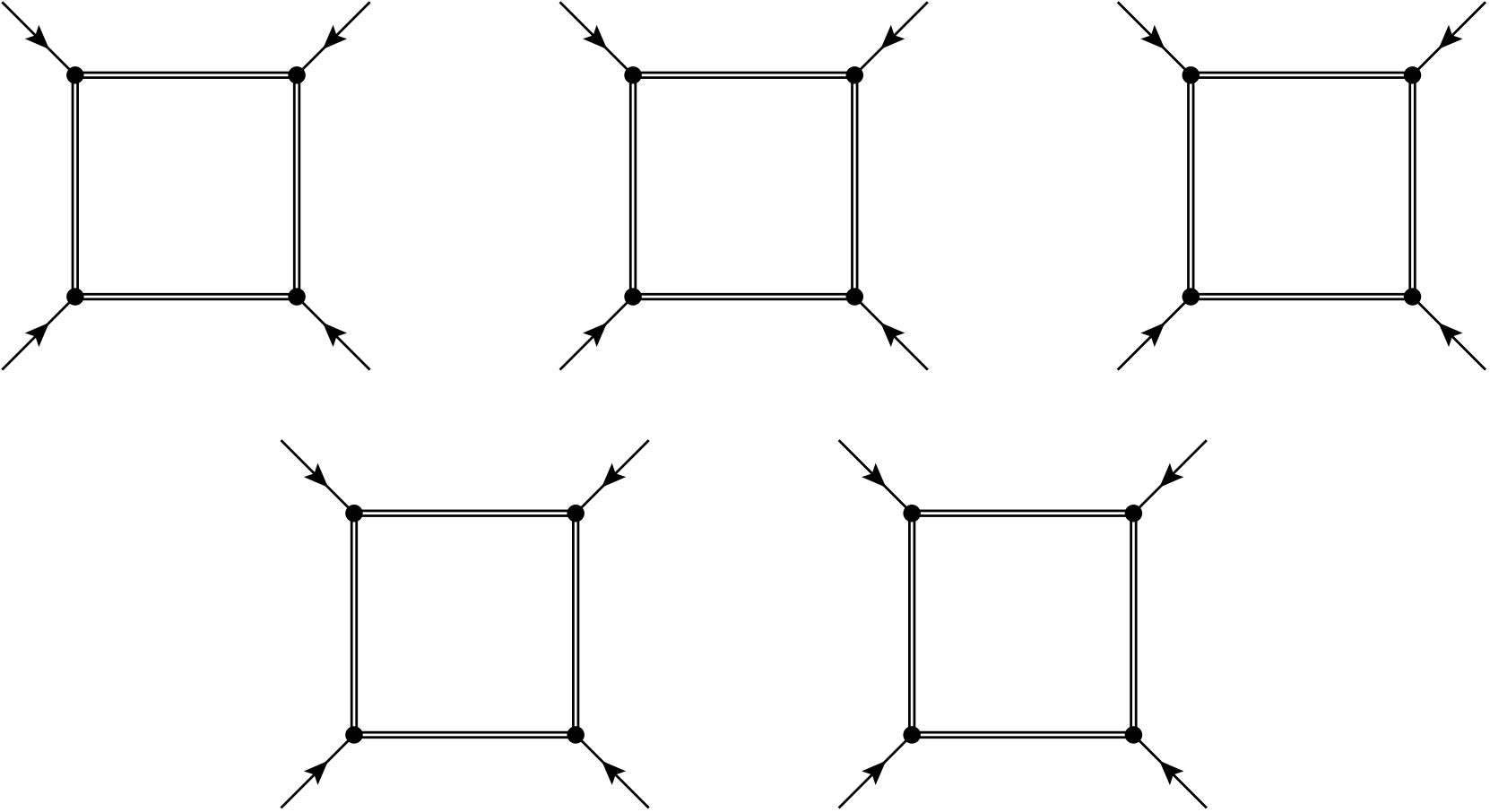}} 
\put(-265,142){\tiny $q_4$}\put(-208,142){\tiny $q_1+q_5$}\put(-230,130){\tiny $m_5$}
\put(-175,142){\tiny $q_3+q_4$}\put(-95,142){\tiny $q_5$}\put(-133,130){\tiny $m_5$}
\put(-75,142){\tiny $q_4+q_5$}\put(0,142){\tiny $q_1$}\put(-37,130){\tiny $m_1$}
\put(-265,75){\tiny $q_3$}\put(-202,75){\tiny $q_2$}\put(-230,82){\tiny $m_3$}
\put(-167,75){\tiny $q_2$}\put(-95,75){\tiny $q_1$}\put(-133,82){\tiny $m_2$}
\put(-72,75){\tiny $q_3$}\put(0,75){\tiny $q_2$}\put(-37,82){\tiny $m_3$}
\put(-65,108){\tiny $m_4$}\put(-10,108){\tiny $m_2$}
\put(-160,108){\tiny $m_3$}\put(-107,108){\tiny $m_1$}
\put(-255,108){\tiny $m_4$}\put(-203,108){\tiny $m_2$}
\put(-208,30){\tiny $m_4$}\put(-155,30){\tiny $m_1$}
\put(-204,65){\tiny $q_4$}\put(-155,65){\tiny $q_5$}\put(-180,55){\tiny $m_5$}
\put(-180,8){\tiny $m_3$}\put(-214,-5){\tiny $q_3$}\put(-155,-5){\tiny $q_1+q_2$}
\put(-112,30){\tiny $m_5$}\put(-57,30){\tiny $m_2$}
\put(-108,65){\tiny $q_5$}\put(-59,65){\tiny $q_1$}\put(-85,55){\tiny $m_1$}\put(-85,8){\tiny $m_4$}\put(-115,-5){\tiny $q_4$}\put(-59,-5){\tiny $q_2+q_3$}
\caption{Five-term relation from pentagon diagram}
\label{fig: 5term}
\end{figure}

Since the canonical integrals for massive boxes are proportional to their Feynman integrals, we deduce a  5-term relation for massive box Feynman integrals. It is known \cite{Geometric1loop} that the latter are expressible as a linear combination of a number of   dilogarithms. The graphical relation above is presumably equivalent to Abel's 5-term equation for the dilogarithm. This functional equation should, in addition, be compatible with the motivic coaction which was computed  in 
 \cite{Tapuskovic1loop}. 

\section{Example of a canonical amplitude in the quaternionic case}  \label{sect: ExamplesQuaternionic}
\subsection{A canonical amplitude for $\varpi^5$}
Let $G$ be the one-loop hexagon diagram with 6 external momenta in $\R^4$, and 6 masses. The  generalised graph Laplacian is
\[ 
\widetilde{\Lambda}_G(\mu, m)= \sum_{i=1}^6  \alpha_i 
\left(
\begin{array}{c|c}
  1   & \mu_i      \\ \hline 
 \overline{\mu}_i  &    \mu_i \overline{\mu_i}+ m_i^2 
\end{array}
\right)
\]
where $\mu_i \in \Qt$ is a quaternionic momentum routing. Its complex adjoint may be represented by 
\[
\chi_{\widetilde{\Lambda}_G(\mu, m)} = 
\sum_{i=1}^6  \alpha_i 
\left(
\begin{array}{cc|cc}
  1 &0   & \lambda_i & \nu_i     \\
 0 &1    &   -\overline{\nu}_i & \overline{\lambda}_i  \\  \hline 
 \overline{\lambda}_i  &  -\nu_i  &  \lambda_i \overline{\lambda}_i+ \nu_i \overline{\nu}_i+m_i^2   &  0  \\
\overline{\nu}_i & \lambda_i &    0  &  \lambda_i \overline{\lambda}_i+ \nu_i \overline{\nu}_i + m_i^2 
\end{array}
\right)
\]
where we write $\mu_i = \lambda_i + \Qj \nu_i$. We find by computer calculation that 
\[ 
\omega^5_G =  60\,  N_G  \, \frac{\Omega_G}{\Xi_G(q,m)^3}
\]
where  $N_G$ is the determinant of the following matrix:
\[
 \begin{pmatrix} 
1 & 1 & 1 & 1& 1& 1 \\
\lambda_1 & \lambda_2 & \lambda_3 & \lambda_4& \lambda_5& \lambda_6 \\
\overline{\lambda}_1 & \overline{\lambda}_2 & \overline{\lambda}_3 & \overline{\lambda}_4& \overline{\lambda}_5& \overline{\lambda}_6 \\
\nu_1 & \nu_2 & \nu_3 & \nu_4& \nu_5& \nu_6 \\
\overline{\nu}_1 & \overline{\nu}_2 & \overline{\nu}_3 & \overline{\nu}_4& \overline{\nu}_5& \overline{\nu}_6 \\
m_1^2 + ||\mu_1||^2 & m_2^2 + ||\mu_2||^2 & m_3^2 + ||\mu_3||^2 &m_4^2 + ||\mu_4||^2 & m_5^2 + ||\mu_5||^2&m_6^2 + ||\mu_6||^2 \\
\end{pmatrix}
\]
where we write $ ||\mu_i||^2 = \lambda_i \overline{\lambda}_i + \nu_i \overline{\nu}_i$. 
Thus the canonical integral associated to the massive hexagon is proportional to the usual Feynman integral (viewed in 6 spacetime dimensions).  Stokes' formula  applied to a massive quaternionic heptagon diagram leads to a 7-term functional equation for these integrals.

\section{Tropical single-valued integration on curves} \label{sect: TropicalSV}

In this section, which is not required for the paper, we provide some motivation for the definition of the generalised graph Laplacian by analogy with  smooth projective algebraic curves with a finite number of punctures. 

\subsection{Periods on smooth algebraic curves}  \label{sect: SVpairing}
Let $X$ be a smooth compact curve over a field $k \subset \C$, and let $\Sigma \subset X(\C)$ denote a finite set of points. 
Consider the    Gysin (residue) sequence in Rham cohomology 
\[
0  \To H^1_{dR}(X) \To H^1_{dR}(X \backslash \Sigma) \overset{\mathrm{Res}}{\To} H_{dR}^0(\Sigma)(-1) \To H_{dR}^2(X) 
\]
and suppose that we are given an element  (`external momentum'):
\[ 
\alpha_q \  \in \     \ker \left( H_{dR}^0(\Sigma)(-1)  \To H_{dR}^2(X) \right)\otimes_k \C\ ,
\]
which may be interpreted as a divisor  of degree $0$ on $X$ with coefficients in $\C$ supported on $\Sigma$.   
By pulling back the above  sequence along the copy of $ \C(-1) \subset H_{dR}^0(\Sigma)(-1)\otimes_k \C$ spanned by $\alpha_q$, we obtain a simple extension over $\C$:
\begin{equation}
0  \To H^1_{dR}(X) \otimes_k \C \To \mathcal{E}_X  \To \C(-1)    \To 0 \ . 
\end{equation} 
Since $F^1 \C(-1) = \C(-1)$, it splits if and only if the sequence
\[
0  \To F^1 H^1_{dR}(X)\otimes_k \C  \To F^1 \mathcal{E}_X \To \C(-1) \To 0  
\]
splits. Such a splitting is given by    an element (`momentum routing')
\begin{equation}
\mu  \quad \in \quad   F^1 \mathcal{E}_X 
\end{equation} 
whose image  in $\C(-1)$ is $\alpha_q$. In other words,  $\mathrm{Res}\,  \mu = \alpha_q$.   Since
$H^0(X, \Omega_X^1 (\log \Sigma)) \overset{\sim}{\rightarrow} \mathrm{gr}_F^1(H^1_{dR}(X\backslash \Sigma))$, we can assume that $\mu$
is represented by a differential of the third kind with only logarithmic poles along $\Sigma$. 
It is uniquely defined only up to adding a linear combination of holomorphic forms:
\[ 
\mu \mapsto \mu + \sum_{i=1}^g  \lambda_i \omega_i
\]
where $\omega_1,\ldots, \omega_g \in 
 H^0(X, \Omega^1_{X}) \cong F^1 H^1_{dR}(X)$ and $\lambda_i \in \C$.

We wish to consider a  Hermitian form $Q$
on $F^1 \mathcal{E}_X$   defined formally by:
\begin{eqnarray} \label{HermitianIntegralPairing}
Q: F^1 \mathcal{E}_X \otimes F^1 \mathcal{E}_X  & \To &  \C   \\ 
\eta_1 \otimes \eta_2 & \mapsto &  \frac{1}{2\pi i} \int_{X} \eta_1 \wedge \overline{\eta_2} \nonumber
\end{eqnarray} 
Alas, the integrals on the right are not always convergent. There are three cases:
\begin{enumerate}
\item $\eta_1=\omega_i$, and  $\eta_2=\omega_j$ are holomorphic. Then  the restiction of $Q$   to $H^1_{dR}(X)$ is given by the convergent integrals 
\[   \omega_i \otimes \omega_j  \mapsto  \frac{1}{2\pi i} \int \omega_i \wedge \overline{\omega}_j \] 
from which one retrieves the Riemann  polarisation form. If $k$ is real, then this actually defines a symmetric quadratic form on $H^1_{dR}(X)$. 
\item $\eta_1 =\omega_i$ is holomorphic, and $\eta_2= \mu$ is of the third kind (or vice-versa). Then the singular  integral
\[
c_i(\mu)=  \frac{1}{2\pi i}\int_X   \omega_i \wedge \overline{\mu}  
\]
is  convergent since it defines a single-valued period of $H^1_{dR}(X\backslash \Sigma)$. To see this, note that 
$F^1 H^1(X, \Sigma)   \overset{\sim}{\rightarrow}   F^1 H^1(X)$ is an isomorphism and since $\omega_i$ vanishes along $\Sigma$ it may canonically be viewed as an element of $H^1(X,\Sigma)$. The integral above  is an instance of the single-valued pairing  \cite[\S6.4]{BrownDupont}:
\[
\mathsf{s}:  H^1_{dR}(X, \Sigma) \otimes H^1_{dR}(X \backslash \Sigma) \To \C
\]

\item Only in the  case $\eta_1 =\eta_2=  \mu$ does one have an ill-defined integral: 
\[
 \frac{1}{2\pi i}\int \mu \wedge \overline{\mu} 
\]
It is singular in  a neighbourhood of a point $\sigma \in \Sigma$, since the integrand may be expressed in local polar coordinates $z = \rho e^{i \theta}$ based at $\sigma$ in the form
\[
\mu \wedge \overline{\mu}  =   \lambda \,  \frac{d\rho}{\rho}  d \theta   + \hbox{regular terms}   \ , 
\]
for some $\lambda \in \C$, which is not `polar-smooth' in the terminology of \cite{BrownDupont}. One can regularise these integrals around each point of $\Sigma$ by a number of methods (see also \cite[Remark 3.22]{BrownDupont}).   The method used  is unimportant for this discussion but the value of the integral will in general depend upon it. In any case  we may simply put
\[
 \frac{1}{2\pi i}\int^{\mathrm{reg}} \mu \wedge \overline{\mu} : =x_m
\]  
where  $x_m$ is any  real and positive number of our choosing. The possibility of choosing $x_m$ relates to the presence of masses in Feynman graphs. 
\end{enumerate}

Thus, with respect to the basis $\omega_1, \ldots, \omega_g, \mu$ of $\mathcal{E}_X$, the  form $Q$ is represented by a Hermitian matrix
\[ 
Q = 
\left(
\begin{array}{ccc|c}
  &   &  &    \vdots \\
  & Q_0   &   & c_i(\mu)  \\
  &   &   & \vdots   \\ \hline
  \cdots  &  \overline{c_i(\mu)}&  \cdots   & x_m
\end{array}
\right)
\]
This is precisely the form of the generalised Laplacian matrix (compare \eqref{LambdaGblockmatrix}).
\subsection{Tropical analogy}
Let $G$ be a connected graph (with no external half edges in the first instance). Choose an orientation on its edges. 
According to \cite{MikhalkinZharkov},   a differential form on $G$ is a linear combination 
\[
\omega= \sum_{e\in E_G} \lambda_e de
\]
where $\lambda_e \in \Z$ such that for all vertices  $v$ one has $\sum_{s(e)= v} \lambda_e = \sum_{t(e)= v} \lambda_e$.  
Thus the space of differential forms $\Omega_G^1$  on $G$ lies in an exact sequence 
\[ 
0 \To \Omega_G^1 \To \Z^{E_G} \To \Z^{V_G} \To \Z \To 0 
\]
and may be  identified with $H_1(G;\Z)$. The inner product on $\Z^{E_G}$  is given by 
\[  \langle d e_i , d e_j \rangle = x_i \delta_{ij} \ ,
\]
and restricts to  a positive definite quadratic form on $\Omega_G^1$. 

Now suppose that $G$ has external half-edges, and that  we are given  an element $\cq \in \mathrm{ker} (\Z^{V_G} \rightarrow \Z) \otimes \C$.
As in \S \ref{sect: DefGraphLaplacianComplex}, we deduce an extension $\widetilde{\Omega}^1_G$  satisfying
\[
0 \To \Omega^1_G \otimes \C \To \widetilde{\Omega}^1_G   \To  \cq \, \C \To 0
\]
and the above inner product  (extended $\C$-linearly)  restricts to a Hermitian form 
\begin{eqnarray}  
 \widetilde{\Omega}^1_G  \otimes_{\R}   \widetilde{\Omega}^1_G  &\To&  \C \\ 
 \eta_1 \otimes \eta_2 & \mapsto &  \langle \eta_1, \overline{\eta}_2 \rangle \nonumber 
\end{eqnarray} 
which is a tropical analogue of \eqref{HermitianIntegralPairing}.  The generalised  graph Laplacian matrix  $\widetilde{\Lambda}_G(q,0)$ for zero masses is precisely the matrix of this Hermitian form with respect to a 
 splitting of $\widetilde{\Omega}^1_G$, which  is given by 
any  differential form 
\[
\nu = \sum_{e} \mu_e de
\]
which satisfies the `residue' equation $\sum_{s(e)= v} \lambda_e -  \sum_{t(e)= v} \lambda_e = \cq_v$ for all vertices $v$. 
 The presence of masses corresponds, as above, to modifying only the value of $\langle \nu , \overline{\nu} \rangle$ in the bottom-right hand corner, and from this perspective, we see that the generalised graph Laplacian \eqref{LambdaGblockmatrix} is a tropical analogue 
  of the regularised single-valued integration pairing discussed in  \S \ref{sect: SVpairing}.

\bibliographystyle{alpha}

\bibliography{biblio}

\end{document}